\documentclass[12pt, a4paper]{article}
\usepackage[utf8]{inputenc}
\usepackage{dcolumn,lscape}%,setspace}
\usepackage{amsmath,longtable,multicol,dcolumn,tabularx,graphicx,amssymb}
\usepackage{exscale,amsthm,multirow,rotating,subcaption}
\usepackage{natbib}
\usepackage[table]{xcolor}
\usepackage{bbm}
\usepackage{hyperref}
\usepackage{tikz,dsfont,float}
\newcommand{\xvec}{\boldsymbol}
\newcommand{\xmat}{\mathbf}
\newcommand{\xset}{\mathds}

\newtheorem{theorem}{Theorem}

\newtheorem{proposition}{Proposition}

\newtheorem{assumption}{Assumption}
\newtheorem{lemma}{Lemma}

\definecolor{col1}{HTML}{EBEBDE} 
\definecolor{col2}{HTML}{777764} 
\definecolor{col3}{HTML}{4F4747} 

\usetikzlibrary{patterns}

\begin{document}

\def\spacingset#1{\renewcommand{\baselinestretch}%
{#1}\small\normalsize} \spacingset{1}

%%%%%%%%%%%%%%%%%%%%%%%%%%%%%%%%%%%%%%%%%%%%%%%%%%%%%%%%%%%%%%%%%%%%%%%%%%%%%%

\title{A Multivariate Spatial and Spatiotemporal ARCH Model}
\author{Philipp Otto\footnote{philipp.otto@ikg.uni-hannover.de}\\
{\small Leibniz University Hannover, Germany}}

\maketitle

\begin{abstract}
This paper introduces a multivariate spatiotemporal autoregressive conditional heteroscedasticity (ARCH) model based on a vec-representation. The model includes instantaneous spatial autoregressive spill-over effects in the conditional variance, as they are usually present in spatial econometric applications. Furthermore, spatial and temporal cross-variable effects are explicitly modelled. We transform the model to a multivariate spatiotemporal autoregressive model using a log-squared transformation and derive a consistent quasi-maximum-likelihood estimator (QMLE). For finite samples and different error distributions, the performance of the QMLE is analysed in a series of Monte-Carlo simulations. In addition, we illustrate the practical usage of the new model with a real-world example. We analyse the monthly real-estate price returns for three different property types in Berlin from 2002 to 2014. We find weak (instantaneous) spatial interactions, while the temporal autoregressive structure in the market risks is of higher importance. Interactions between the different property types only occur in the temporally lagged variables. Thus, we see mainly temporal volatility clusters and weak spatial volatility spill-overs. 
\end{abstract}

\noindent%
{\it Keywords:}  Conditional heteroscedasticity, multivariate spatiotemporal data, QML estimator, real-estate prices, volatility clustering

\spacingset{1.45} % DON'T change the spacing!

% \tableofcontents

\section{Introduction}\label{sec:introduction}

In general, spatiotemporal processes can be represented as multivariate time series. However, when analysing spatial and spatiotemporal data, one has to account for one key difference compared to multivariate time series. Due to their spatial nature, geographical proximity between the observations induces instantaneous interactions between them. This is commonly known as Tobler's first law of geography: ``everything is related to everything, but near things are more related than distant things'' (\citealt{Tobler70}). This does not only apply to the mean behaviour of the data, but also their variance. Thus, spatiotemporal models should always allow for instantaneous spatial interactions.

In this paper, we introduce a multivariate spatial and spatiotemporal autoregressive conditional heteroscedasticity (spatial ARCH, briefly spARCH) model. Using a vector representation, we extend the spatial ARCH models of \cite{Otto18_spARCH, sato2021spatial,OttoSchmid19_arxiv_unified} to multivariate and spatiotemporal data. In that sense, the approach follows the same logic as classical time-series vec-ARCH models (cf. \citealt{engle1995multivariate}). Hence, we call the new multivariate, spatiotemporal ARCH process vec-spARCH. All these approaches trace back to the seminal papers of \cite{Engle82} and \cite{Bollerslev86}. In contrast to previous multivariate spatiotemporal GARCH models (e.g., \citealt{Borovkova12}), we allow for instantaneous dependence over space at the same time point, which is important for spatial models. Thus, our multivariate vec-spARCH model can also be applied for purely spatial data, when there are spatial volatility clusters (i.e., clustered regions of high/low volatilities).  Moreover, it is worth noting that, using this vec-representation, the above-mentioned spatial GARCH proposed by \cite{OttoSchmid19_arxiv_unified} nests multivariate models, such that their results can also directly be applied. Alternative models that allow instantaneous spatial interactions in the variance are spatial stochastic volatility models, as proposed by \cite{tacspinar2021bayesian}.

The vec-spARCH process distinguishes between three different effects: (1a) instantaneous spatial effects of the same variables, (1b) instantaneous cross-variable spatial effects, (2a) temporal autoregressive effects of the same variables, (2b) cross-variable temporal autoregressive effects, and (3) variable-specific unconditional volatility levels. Each of the effects is described by a parameter matrix or vector, for which we derive a quasi maximum-likelihood (QML) estimator. For estimation, a logarithmic transformation of the vec-spARCH is applied (cf. \citealt{Robinson:2009}), such that the model can be represented as multivariate spatiotemporal autoregressive model of the transformed quantity. The asymptotic consistency of the QML estimator has been shown by \cite{yang2017identification} for a multivariate spatial autoregressive model (i.e., without temporal dimension) and by \cite{Yu08} for a univariate spatiotemporal autoregressive process. Under certain regularity conditions, which are commonly used in spatial econometrics, we show the identifiability and consistency of the estimators.  

In practice, spatiotemporal ARCH models are particularly important, because an ARCH error process can also account for variation due to latent factors. In particular, for small spatial units, it is often difficult to quantify influential factors with the same spatial resolution. For instance, the average income of households in small spatial units, e.g. postal-code areas, does not necessarily reflects the economic power of these particular units, because people's daily cycles usually span across multiple small spatial units. That is, people do not necessarily live where they work or spend most of their time. In such cases, spatial and spatiotemporal ARCH models are important error distributions of any model to account for unobservable factors.

The remainder of the paper is structured as follows. Firstly, we introduce the multivariate modelling framework and discuss how the model can transformed to a regular univariate spatiotemporal process. Further, we derive the Gaussian logarithmic likelihood and show the asymptotic consistency for the QML estimator under several regularity assumption that are often met in spatial econometrics. Secondly, we analyse the finite-sample performance of the proposed estimator for several model specifications and two different error distributions, namely standard normal and heavy-tailed errors ($t_3$-distributed). Thirdly, a real-world application is presented, for which we show that it is important to account for instantaneous spatial interactions and cross-variable correlations. To be precise, Berlin real-estate prices of three different property types are analysed and we find weak spatial interactions, even though they are dominated by the temporal effects. These interdependencies are more pronounced when the spatial units and time intervals are small. Finally, Section \ref{sec:conclusion} concludes the paper with a summary and brief outlook to future research and potential fields of application.

\section{Multivariate Spatiotemporal ARCH Model}\label{sec:models}

In spatial statistics/econometrics, autoregressive spill-over effects are instantaneous. That is, no time lag is required for shocks to affect neighbouring locations.  Instead, we assume that the conditional variance can vary over space depending on the realised variance at adjacent locations. This results in spatial clusters of high and low variances. For previous univariate or multivariate spatiotemporal GARCH models, such as proposed by \cite{Borovkova12,holleland2020stationary}, spatial spill-overs could only occur after one time instance. In other words, the conditional variance at each locations depends on the past squared observations at the same location and its neighbours, but not on their neighbouring locations at the same time point. This is the fundamental difference between multivariate time series models covering spatiotemporal data and approaches from spatial statistics or econometrics.

\subsection{Model specification}\label{sec:spGARCH}

Assume that $\left\{\xvec{Y}_t(\xvec{s}) \in \xset{R}^p: \xvec{s} \in D_{\xvec{s}} \subset \xset{R}^q, t \in \xset{Z} \right\}$  is a $p$-variate spatiotemporal stochastic process in a $q$-dimensional space $D_{\xvec{s}}$ with positive volume (cf. \citealt{Cressie11}). More precisely, the process is observed at $n$ different locations $\xvec{s}_1, \ldots, \xvec{s}_n$, i.e., at each location $\xvec{s}_i$ and time point $t$ we observe a vector $\xvec{Y}_t(\xvec{s}_i) = (Y_{1,t}(\xvec{s}_i), \ldots, Y_{p,t}(\xvec{s}_i))'$. Moreover, let $\xvec{Y}_{j,t} = (Y_{j,t}(\xvec{s}_1), \ldots, Y_{j,t}(\xvec{s}_n))'$ the vector of the $j$-th characteristic at all locations and $\xmat{Y}_t = (\xvec{Y}_{1,t}, \ldots, \xvec{Y}_{p,t})$ be an $n \times p$ matrix of all observations $\left(Y_{j,t}\left(\xvec{s}_i\right)\right)_{i = 1, \ldots, n, j = 1, \ldots, p}$ at time point $t$. Suppose that the process is observed for $t = 1, \ldots, T$. It is worth mentioning that a multivariate spatial log-ARCH model is present if $T = 1$ and a classical time-series log-ARCH models are also nested if $D_{\xvec{s}}$ is a singleton (i.e., $n = 1$).

Univariate spatial and spatiotemporal ARCH models have been introduced by \cite{Otto18_spARCH} and \cite{Sato17}.  Moreover, \cite{OttoSchmid19_arxiv_unified} generalised the model in a unified framework nesting spatial GARCH, E-GARCH, and Log-GARCH models. In this paper, we follow the idea of the symmetric spatial log-GARCH model of \cite{sato2021spatial}, which includes elements of GARCH and E-GARCH models, but does not coincide with one or the other even if $D_{\xvec{s}}$ consist of only a single location (i.e., the classical time series case). More precisely, the link function between the spatial volatility term is logarithmic like for E-GARCH models, while the volatility term depends on some transformation of the squared observed process (similar to GARCH models).  In contrast to time-series models, in which the temporal lag is clearly defined by the past observations and future observations are not allowed to influence the current observation, there are complex interdependencies in spatial settings and there is no causal relation between the observations anymore. For instance, with two locations $\xvec{s}_1$ and $\xvec{s}_2$ (i.e., $n = 2$), location $\xvec{s}_1$ would influence $\xvec{s}_2$ at each time point and vice versa. In a time series context, this would correspond to a simultaneous influence from future and past values. Thus, for direct generalisation of GARCH or E-GARCH models like in \cite{Otto18_spARCH,OttoSchmid19_arxiv_unified}, difficult assumptions for the existence or invertibility of the process are required in the general case. In addition, existing software could directly be used with some adaptations for the spatiotemporal case (see \citealt{Otto19_RJournal}).

The multivariate spatiotemporal ARCH model (vec-spARCH) is given by
\begin{equation} \label{eq:initial2}
	\left\{ 
	\begin{array}{ccl}
 	\xmat{Y}_t & = & \xmat{H}_t^{1/2} \xmat{\Xi}_t  \qquad \text{with} \\
    \ln\xmat{H}_t & = & \xmat{A} + \xmat{W}  \, \ln\xmat{Y}_t^{(2)} \,  \xmat{\Psi}  +  \ln\xmat{Y}_{t-1}^{(2)} \,  \xmat{\Pi} \, ,
    \end{array} \right.
\end{equation}
where $\ln\xmat{Y}_t^{(2)}$ denotes the matrix of squared observations $\left(\ln Y_{j,t}^2(\xvec{s}_i)\right)_{i = 1, \ldots, n, j = 1, \ldots, p}$, and $\ln\xmat{H}_t$ is the matrix of all $\ln h_{j,t}(\xvec{s}_i)$ with $i = 1, ..., n$ rows and $j = 1, ..., p$ columns. This matrix is the spatial equivalent of the conditional volatility (see \citealt{Otto19_statpapers}). Moreover, the matrix of disturbances is denoted by $\xmat{\Xi}_t = (\xvec{\varepsilon}_{1,t}, \ldots, \xvec{\varepsilon}_{p,t})$ with independent and identically distributed random vectors $\xvec{\varepsilon}_{j,t}$ with $E(\xvec{\varepsilon}_{j,t}) = \xvec{0}$ and $Cov(\xvec{\varepsilon}_{j,t}) = \xmat{I}$. The weight $n \times n$ matrix $\xmat{W}$ defines the spatial dependence structure, i.e., which locations are considered to be adjacent. Moreover, the cross-variable spatial effects are represented by the off-diagonal elements of $\xmat{\Psi}$, and the temporally lagged cross-variable effects are given by the off-diagonal elements of $\xmat{\Pi}$. Both matrices have dimension $p \times p$. In addition, the own-variable spatial and temporal autoregressive ARCH effects are summarised by the diagonal entries of $\xmat{\Psi}$ and $\xmat{\Pi}$, respectively.

Analogue to multivariate vec-ARCH time-series model of \cite{engle1995multivariate}, we can rewrite \eqref{eq:initial2} to get the vectorised form
\begin{equation}\label{eq:vec_representation}
\ln vec(\xmat{H}_t) = vec(\xmat{A}) + (\xmat{\Psi}' \otimes \xmat{W}) \ln vec(\xmat{Y}_t^{(2)}) +  (\xmat{\Pi}' \otimes \xmat{I}) \ln vec(\xmat{Y}_{t-1}^{(2)}) \, .
\end{equation}
The Kronecker product is denoted by $\otimes$. Interestingly, using such vec-representation, one can see that the multivariate ARCH model is a special case of (univariate) $np$-dimensional spatial GARCH models with a weight matrix $\xmat{\Psi}' \otimes \xmat{W}$. Thus, also spatial GARCH and E-GARCH models can be constructed in the same manner and all results of \cite{OttoSchmid19_arxiv_unified} can directly be applied. However, this will not be the focus of this paper.

Moreover, the multivariate spatiotemporal ARCH model can be written as multivariate spatiotemporal autoregressive process by applying a log-squared transformation, 
\begin{equation*}\label{eq:vec_representation}
\ln \xmat{Y}_t^{(2)} = \ln \xmat{H}_t + \ln \xmat{\Xi}_t^{(2)} \, .
\end{equation*}
Then, we get that
\begin{equation*}
\ln \xmat{Y}_t^{(2)} = \xmat{A} + \xmat{W} \ln \xmat{Y}_t^{(2)} \xmat{\Psi} +   \ln\xmat{Y}_{t-1}^{(2)} \,  \xmat{\Pi} + \ln \xmat{\Xi}_t^{(2)} \, .
\end{equation*}
With $\xmat{U}_t = \ln \xmat{\Xi}_t^{(2)} - E\left(\ln \xmat{\Xi}_t^{(2)}\right)$ and $\tilde{\xmat{A}} = \xmat{A} + E\left(\ln \xmat{\Xi}_t^{(2)}\right)$, the model can be rewritten as
\begin{equation*}
\ln \xmat{Y}_t^{(2)} = \tilde{\xmat{A}} + \xmat{W} \ln \xmat{Y}_t^{(2)} \xmat{\Psi} +  \ln\xmat{Y}_{t-1}^{(2)} \,  \xmat{\Pi} + \xmat{U}_t \, .
\end{equation*}
Hence, the vec-spARCH model coincides with a multivariate spatiotemporal autoregressive process of the log-squared transformed process $\ln \xmat{Y}_t^{(2)}$. For the multivariate but purely spatial case, \cite{yang2017identification} has derived conditions for identification and the consistency and asymptotic normality of a QML estimator. Furthermore, \cite{Yu08} derive asymptotic results for a QML estimator of spatiotemporal but univariate process when both $n$ and $T$ are large. We combine these two results to propose a QML estimator for the spatiotemporal, multivariate ARCH model.

Assuming a standard normal error matrix $\xmat{\Xi}_t$, $E\left(\ln \xmat{\Xi}_t^{(2)}\right)$ is the expectation of a log-Gamma distribution, i.e., $E\left(\ln \xmat{\Xi}_t^{(2)}\right) = \gamma - \log(2) \approx -1.27$. Then, $\xmat{A}$ can be determined from $\tilde{\xmat{A}}$, which facilitates the interpretation. With $\xmat{S}_{np} = \xmat{I} - \xmat{\Psi}'\otimes\xmat{W}$, we can derive the sample log-likelihood for the spatiotemporal case with $T$ time points, i.e.,
\begin{eqnarray*}\label{eq:LL}
\ln \mathcal{L}( \xmat{A}, \xmat{\Psi}, \xmat{\Pi} | \xmat{Y}_0) & = & - \frac{Tnp}{2} \ln (2\pi) + \frac{n\ln\sigma^2_u}{2p} + \frac{T}{np} \ln |\xmat{S}_{np}| \\ 
                              &   & - \frac{1}{2np\sigma^2_u} \sum_{t = 1}^{T} \left[ \xmat{S}_{np} \ln vec(\xmat{Y}_t^{(2)}) - vec(\tilde{\xmat{A}}) - (\xmat{I} \otimes \ln \xmat{Y}^{(2)}_{t-1}) vec(\xmat{\Pi}) \right]' \\
                              &   & \qquad\qquad\qquad \times \left[ \xmat{S}_{np} \ln vec(\xmat{Y}_t^{(2)}) - vec(\tilde{\xmat{A}}) - (\xmat{I} \otimes \ln \xmat{Y}^{(2)}_{t-1}) vec(\xmat{\Pi}) \right]\, ,
\end{eqnarray*}
where $\sigma^2_u$ is the variance of the transformed errors $\xmat{U}_t$, which is known quantity in our case (otherwise ${\xmat{A}}$ would not be identifiable). Furthermore, for standard normal $\xmat{\Xi}_t$, we get $\sigma^2_u = \psi(1/2) \approx 4.93$, where $\psi$ denotes the trigamma function. It is worth mentioning that we derived the log-likelihood for multivariate Gaussian errors $\xmat{U}_t$, which are in fact skewed because of the logarithmic transformation. In the following Section \ref{sec:assumptions}, however, we suppose much weaker conditions for the moments of $\xmat{U}_t$, which are fulfilled in the case of standard normal $\xmat{\Xi}_t$, for instance. Furthermore, let $\ddot{\xvec{Y}}_t = \ln vec(\xmat{Y}_t^{(2)})$ for an easier notation. With $E(\ddot{\xvec{Y}}_t) = \xmat{S}_{np0}^{-1} \left(vec(\tilde{\xvec{A}}_0) + \xmat{\Pi}_0' \otimes \ddot{\xvec{Y}}_{t-1} \right)$, we get the expected log-likelihood as
\begin{eqnarray*}\label{eq:ELL}
E(\ln \mathcal{L}( \xmat{A}, \xmat{\Psi}, \xmat{\Pi} | \xmat{Y}_0)) & = & - \frac{Tnp}{2} \ln (2\pi) + \frac{n\ln\sigma^2_u}{2p} + \frac{T}{np} \ln |\xmat{S}_{np}| \\ 
                              &   & - \quad \frac{1}{2np\sigma^2_u} \sum_{t = 1}^{T} \left[ \xmat{S}_{np} \xmat{S}_{np0}^{-1} (vec(\tilde{\xmat{A}}_0 - \tilde{\xmat{A}}) + ((\xmat{\Pi}_0' - \xmat{\Pi}') \otimes \xmat{I}) \ddot{\xvec{Y}}_{t-1}) \right]' \\
                              &   & \qquad\qquad\qquad \times \left[ \xmat{S}_{np} \xmat{S}_{np0}^{-1} ( vec(\tilde{\xmat{A}}_0 - \tilde{\xmat{A}}) + ((\xmat{\Pi}_0' - \xmat{\Pi}') \otimes \xmat{I}) \ddot{\xvec{Y}}_{t-1})  \right] \\
                             &   & - \quad \frac{T}{2np\sigma^2_u} tr\left( \xmat{S}_{np} \xmat{S}_{np0}^{-1} \xmat{S}_{np0}^{'-1} \xmat{S}_{np}'\right) \, .\\
\end{eqnarray*}

\subsection{Assumptions and parameter space}\label{sec:assumptions}

Below, we discuss important model assumptions that are also needed to derive the asymptotic consistency of the QML estimators.  

\begin{assumption}\label{ass:1}
	Suppose that each element of $\xmat{\Xi}_t$ is not equal to zero with probability one for all $t = 1, \ldots, T$.
	%Suppose that $P(Y_{t,i,j} = 0) = 0$ for all $t = 1, \ldots, T$, $i = 1, \ldots, n$, and $j = 1, \ldots, p$. 
\end{assumption}

To be able to apply the log-squared transformation of the observed process, we must assume that the response is not equal to zero with probability one. This is the case for any continuous error process $\xmat{\Xi}_t$. In practice, due to missing values, there is sometimes an excess of zeros. In this case, often a small number is added to the zero values, such that the logarithmic transformation gets feasible (see, e.g., \citealt{Francq11}). If there are zero values with a probability larger than zero, \cite{sucarrat2018estimation} proposed an expectation-maximisation algorithm for estimation in the time-series case. This would be an interesting extension for future research. Further, we need some basic assumptions on the transformed error process $\xmat{U}_t$ to apply the results of \cite{yang2017identification} and \cite{Yu08}.

\begin{assumption}\label{ass:2}
	Assume that each row $j = 1, \ldots, p$ of $\xmat{U}_t$, say $U_{t,j}$, is a random vector with zero mean and covariance $\sigma_u^2 \xmat{I}$ that is i.i.d. across time. Additionally suppose that $E(|u_{t,ik}u_{t,il}u_{t,ip}u_{t,iq}|^{1/\delta}) < \infty$ for all $i = 1, \ldots, n$, $t = 1, \ldots, T$ and $k,l,p,q = 1, \ldots, p$ and some $\delta > 0$. 
\end{assumption}

%+ stability across time?? % http://pareto.uab.es/lgambetti/Lecture2M.pdf
%% see Assumption 6 Yu 2008
%+ examples or explanations when this is the case

Moreover, the parameter space needs to be compact, as formulated in the following assumption, to prove the uniform convergence of the log-likelihood function.

\begin{assumption}\label{ass:3}
	The parameter spaces for $\tilde{\xmat{A}}$, $\xmat{\Psi}$ and $\xmat{\Pi}$ are compact sets and all parameters in their interior generate a stable process. Moreover, the data-generating parameters $\tilde{\xmat{A}}_0$, $\xmat{\Psi}_0$ and $\xmat{\Pi}_0$ are in the interior of corresponding parameter space.
\end{assumption}

The key question of this assumption is the stability of the process. One could rewrite the model as 
\begin{eqnarray*}\label{eq:stable_series}
\ddot{\xvec{Y}}_t & = & \xmat{S}_{np}^{-1} vec(\tilde{\xmat{A}}) + \xmat{S}_{np}^{-1}(\xmat{\Pi}' \otimes \xmat{I})\ddot{\xvec{Y}}_{t-1} + vec(\xmat{U}_t) \\
				  & = & \xmat{S}_{np}^{-1} vec(\tilde{\xmat{A}}) + \xmat{S}_{np}^{-1}(\xmat{\Pi}' \otimes \xmat{I})
				  \left[ \xmat{S}_{np}^{-1} vec(\tilde{\xmat{A}}) + \xmat{S}_{np}^{-1}(\xmat{\Pi}' \otimes \xmat{I})\ddot{\xvec{Y}}_{t-2} + vec(\xmat{U}_{t-1})\right] 
				  + vec(\xmat{U}_t) \\	
				  & \vdots &  \\	
				  & = & (\xmat{I} + \xmat{S}_{np}^{-1}(\xmat{\Pi}' \otimes \xmat{I}) + \ldots + (\xmat{S}_{np}^{-1}(\xmat{\Pi}' \otimes \xmat{I}))^j) vec(\tilde{\xmat{A}}) + (\xmat{S}_{np}^{-1}(\xmat{\Pi}' \otimes \xmat{I}))^j \ddot{\xvec{Y}}_{t-j} \\
				   &   & \quad + \sum_{i = 0}^{j-1} (\xmat{S}_{np}^{-1}(\xmat{\Pi}' \otimes \xmat{I}))^i vec(\xmat{U}_{t-i}) \\	
\end{eqnarray*}
Hence, the stability of the process does not only depend on the temporal parameter matrix $\xmat{\Pi}$ but also on the weight matrix $\xmat{W}$ (via $\xmat{S}_{np}^{-1}$). If the above series converges, we get a stable and stationary process. 

\begin{proposition}\label{prop:stability}
	If all eigenvalues of $\xmat{S}_{np}^{-1}(\xmat{\Pi}' \otimes \xmat{I})$ are smaller than one, the multivariate spatiotemporal ARCH process is stable across time.
\end{proposition}

Note that each stable spatiotemporal ARCH process is also weakly stationary. Furthermore, the boundary region of $\xmat{\Psi}$ where $|\xmat{S}_{np}| = 0$ can be problematic in practice. However, as long as the true parameter $\xmat{\Psi}_0$ is bounded away from this region, the maximisation algorithm will not get to these boundaries with a large probability (see also \citealt{yang2017identification}).  

% note that covariance matrix of the errors is always invertible, because we have to assume homoscedastic errors with variance one (i.e., I)

\begin{assumption}\label{ass:4}
	The row and column sums of $\xmat{W}$ in absolute values are uniformly bounded in $n$. Moreover, $\xmat{S}_{np}$ is invertible for all possible matrices $\xmat{\Psi}$ in the parameter space and $\xmat{S}_{np}^{-1}$ is uniformly bounded in absolute row and column sums. % , uniformly in $\xmat{\Psi}$
\end{assumption}

Assumption \ref{ass:4} is classical in spatial statistics to obtain a stable process across space (cf. \citealt{yang2017identification,Kelejian98,Lee04}). Here, we could adopt the assumption as formulated in \cite{yang2017identification} for multivariate spatial autoregressive models. Together with Proposition \ref{prop:stability}, we obtain a stable process across space and time. In practice, the spatial weight matrices often standardised to meet these regularity conditions, e.g. the most widely adopted row-wise standardisation.

% + explanation: this assumption is needed to obtain stability across space, Kelejian and Prucha 1998, Lee 2004

\begin{assumption}\label{ass:5}
	Let $n$ be a nondecreasing function of $T$ and $T \to \infty$.
\end{assumption}

Assumption \ref{ass:5} implies that $n,T \to \infty$ simultaneously.

\subsection{Consistency of the QML estimator}\label{sec:qml} % asymptotic normality???

Due to the presence of endogenous variables, i.e., the instantaneous spatial interactions, the identification of spatial models is generally more difficult than in the strict time-series case, where all spatiotemporal interaction may only occur after one time lag. Thus, we initially focus on the identification of the parameters which is needed for the asymptotic consistency of the QML estimator in the following Theorem \ref{th:consistency}. Since the identification is inherent with the spatial dimension of the model, we could follow the same strategy as in \cite{yang2017identification} for multivariate spatial autoregressive models. The identification is based on the information inequality, as proposed by \cite{Rothenberg71}.

\begin{proposition}\label{prop:identification}
	If the Assumptions \ref{ass:1}-\ref{ass:5} are fulfilled, then $\tilde{\xmat{A}}_0$, $\xmat{\Psi}_0$ and $\xmat{\Pi}_0$ are uniquely identifiable.
\end{proposition}

For the identification, we make use of the fact that the spatial dependence is constant across time and the temporal dependence is constant for all spatial locations. If either of them varies in space or time, further identifying information would be needed. Moreover, in contrast to \cite{yang2017identification}, the errors are uncorrelated by definition and the error variance is supposed to be known. The assumption of an uncorrelated error process is essential for GARCH models for identification of the parameters in the conditional volatility equation, i.e., the so-called GARCH effects. Moreover, the assumption of a known error variance $\sigma^2_u$ is of course restrictive (see also \citealt{Francq11,Brockwell06}) and it is often difficult to choose an appropriate value. For time series, ex-post scale adjustments have been proposed to circumvent this assumption (see \citealt{bauwens2010general,sucarrat2016estimation}). In this paper, however, we follow the classical approach and point to future research for these ex-post scale adjustments. Moreover, for the purely spatial case with $T = 1$, $\tilde{\xmat{A}}_0$ must be constant across space to be identifiable.

To estimate the parameters, we propose a QML estimator based on the log-likelihood function given by \eqref{eq:LL}. That is, the parameters $\xvec{\vartheta} = (vec(\tilde{\xmat{A}})', vec(\xmat{\Psi})', vec(\xmat{\Pi})')$ can be estimated by
\begin{eqnarray*}\label{eq:LL}
\hat{\xvec{\vartheta}}_{nT} = \arg\max_{\xvec{\vartheta} \in \Theta} \ln \mathcal{L}( \xmat{A}, \xmat{\Psi}, \xmat{\Pi} | \xmat{Y}_0), 
\end{eqnarray*}
where $\Theta$ is the parameter space that fulfils Assumption \ref{ass:3}. It is worth noting that we need to condition on the observed vector at $t = 0$, $\xvec{Y}_0$, because of the temporal autoregressive structure. The asymptotic consistency of this QML estimator is summarised in the following theorem.
 
\begin{theorem}\label{th:consistency}
	Under Assumptions \ref{ass:1}-\ref{ass:5}, $\xvec{\vartheta}_0 = (vec(\tilde{\xmat{A}}_0)', vec(\xmat{\Psi}_0)', vec(\xmat{\Pi}_0)')$ can be uniquely identified and $\hat{\xvec{\vartheta}}_{nT} \overset{p}{\to} \xvec{\vartheta}_0$.
\end{theorem}

\section{Monte Carlo Simulations}

In the following section, we present the results of a series of simulations on the consistency of the parameters for finite samples. To give a first visual impression, we display a bivariate spatial ARCH process ($T = 1$, $n = 900$, Rook's continuity matrix) with and without spatial cross-correlation in Figure \ref{fig:simulations}. For both examples, the spatial ARCH effects are equal to 0.5, a moderate level of spatial dependence. Therefore, spatial volatilities clusters can be seen in both cases. They are indicated by a higher variance, that is, more intensely coloured pixels, whereas clusters of low variance are close to zero indicated by evenly grey coloured pixels. Now, for the case with a cross-correlation of 0.35 (top panels), these clusters are aligned across the variables, while they are located at different positions in the lower panels with zero cross-correlation. 

\begin{figure}
  \centering
  \includegraphics[width=0.6\textwidth]{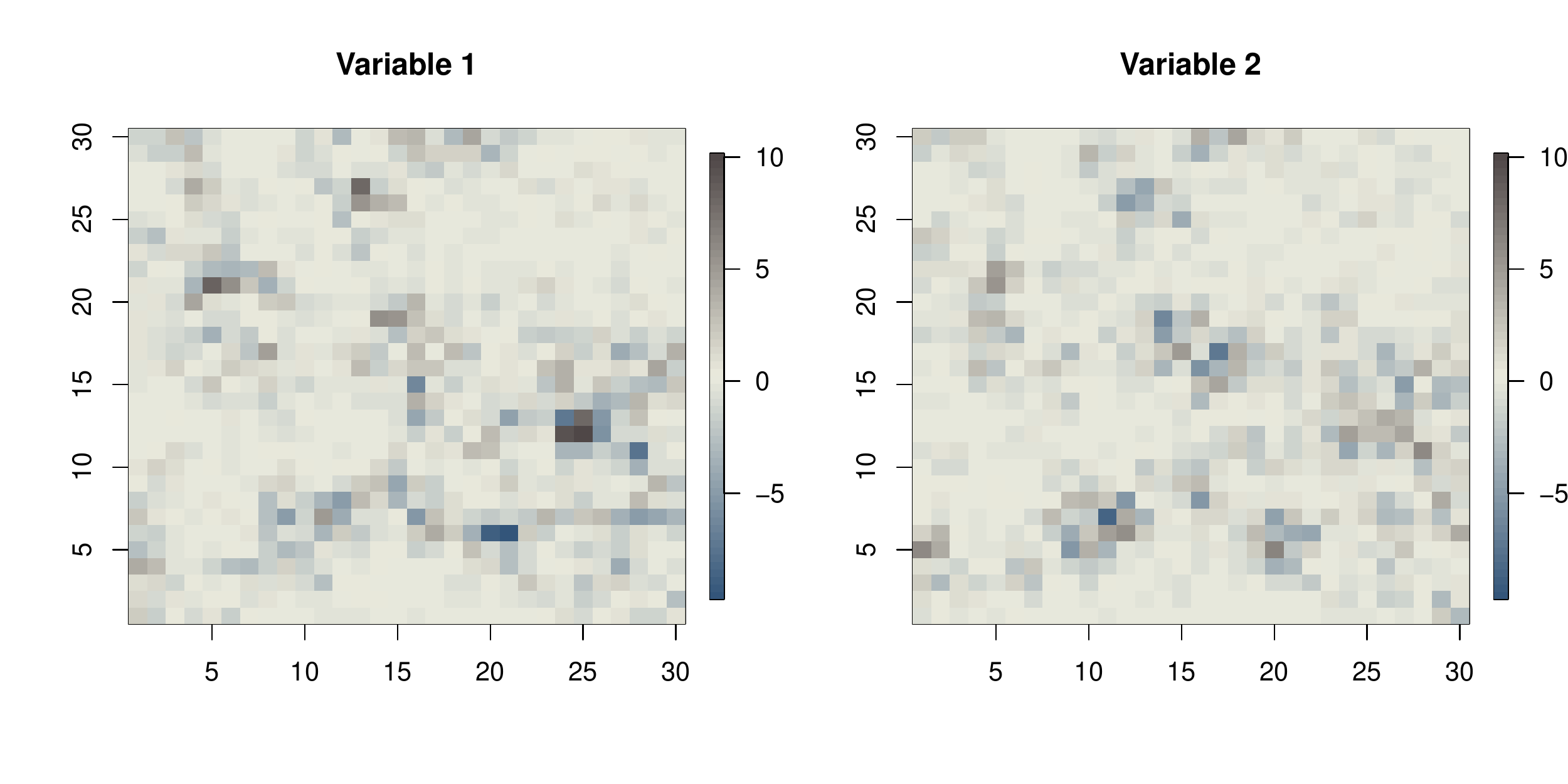}\\
  \includegraphics[width=0.6\textwidth]{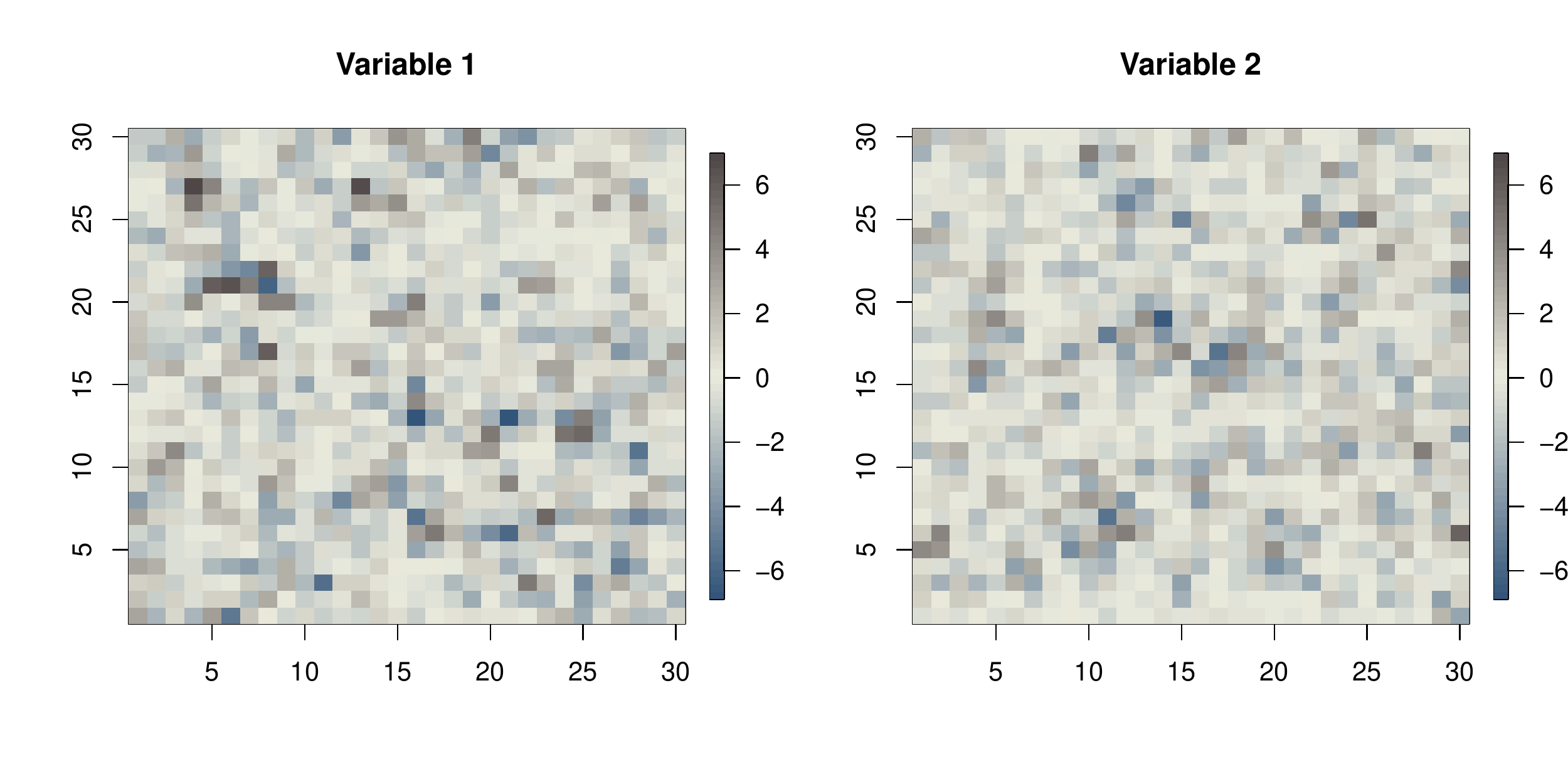}
  \caption{Simulated random fields (first row: high cross correlation $\psi_{12} = \psi_{21} = 0.35$; second row: no cross correlation $\psi_{12} = \psi_{21} = 0$). The spatial ARCH coefficients are identical for all components and both settings, i.e., $\psi_{11} = \psi_{22} = 0.5$.}\label{fig:simulations}
\end{figure}

In our Monte Carlo simulation study, we simulated three different bivariate models (A, B, C) with two different error distributions (standard normal and $t_3$) with 1000 replications. For each combination, we successively increased the size of the spatial field $n \in \{25, 49, 100\}$ and the length of the time series $T \in \{30, 100, 200\}$. We simulated the process on a two-dimensional grid as visualised in Figure \ref{fig:simulations} and the spatial weight matrix was chosen as row-standardised Queen's contiguity matrix. The data-generating parameters of the three considered models are as follows:
\begin{enumerate}
\item[(A)] Spatiotemporal model with a weak spatial cross-correlation: $\xmat{A}_0 = \xvec{1}_n\xvec{1}_p'$, $\xmat{\Psi}_0 = \left( \begin{matrix} 0.5 & 0.1 \\ 0.1 & 0.5 \end{matrix} \right)$, and $\xmat{\Pi}_0 = \left( \begin{matrix} 0.3 & 0 \\ 0 & 0.3 \end{matrix} \right)$
\item[(B)] Spatiotemporal model without temporal dependence, but the same spatial dependence like for Model A: $\xmat{A}_0 = \xvec{1}_n\xvec{1}_p'$, $\xmat{\Psi}_0 = \left( \begin{matrix} 0.5 & 0.1 \\ 0.1 & 0.5 \end{matrix} \right)$, and $\xmat{\Pi}_0 = \left( \begin{matrix} 0 & 0 \\ 0 & 0 \end{matrix} \right)$
\item[(C)] Spatiotemporal model with pronounced cross-correlation and weak spatial correlation, same temporal autocorrelation like for Model A: $\xmat{A}_0 = \xvec{1}_n\xvec{1}_p'$, $\xmat{\Psi}_0 = \left( \begin{matrix} 0.2 & 0.4 \\ 0.4 & 0.2 \end{matrix} \right)$, and $\xmat{\Pi}_0 = \left( \begin{matrix} 0.3 & 0 \\ 0 & 0.3 \end{matrix} \right)$.
\end{enumerate}

For the first simulated model, i.e., Model A with standard normal errors, the parameter estimates are depicted as a series of boxplots for the three increasing sizes $(n,T)' \in \{(25,30)', (49,100)', (100,200)'\}$ in Figure \ref{fig:MC_results}. In all cases, the asymptotic consistency of the QML estimator can be seen, because the boxplots are getting more centred around zero. Moreover, we see the typical bias of the QML estimators for small spatial fields, which rapidly vanishes with an increasing sample size. The same behaviour can be observed for all other settings and error distributions. The average bias and the root-mean-square errors (RMSE) are reported in Tables \ref{table:MC_bias} and \ref{table:MC_rmse}, respectively. Both the absolute values of the bias and the RMSE are approaching zero if $n$ and $T$ are increasing.

\begin{figure}
  \centering
  \includegraphics[width=0.6\textwidth]{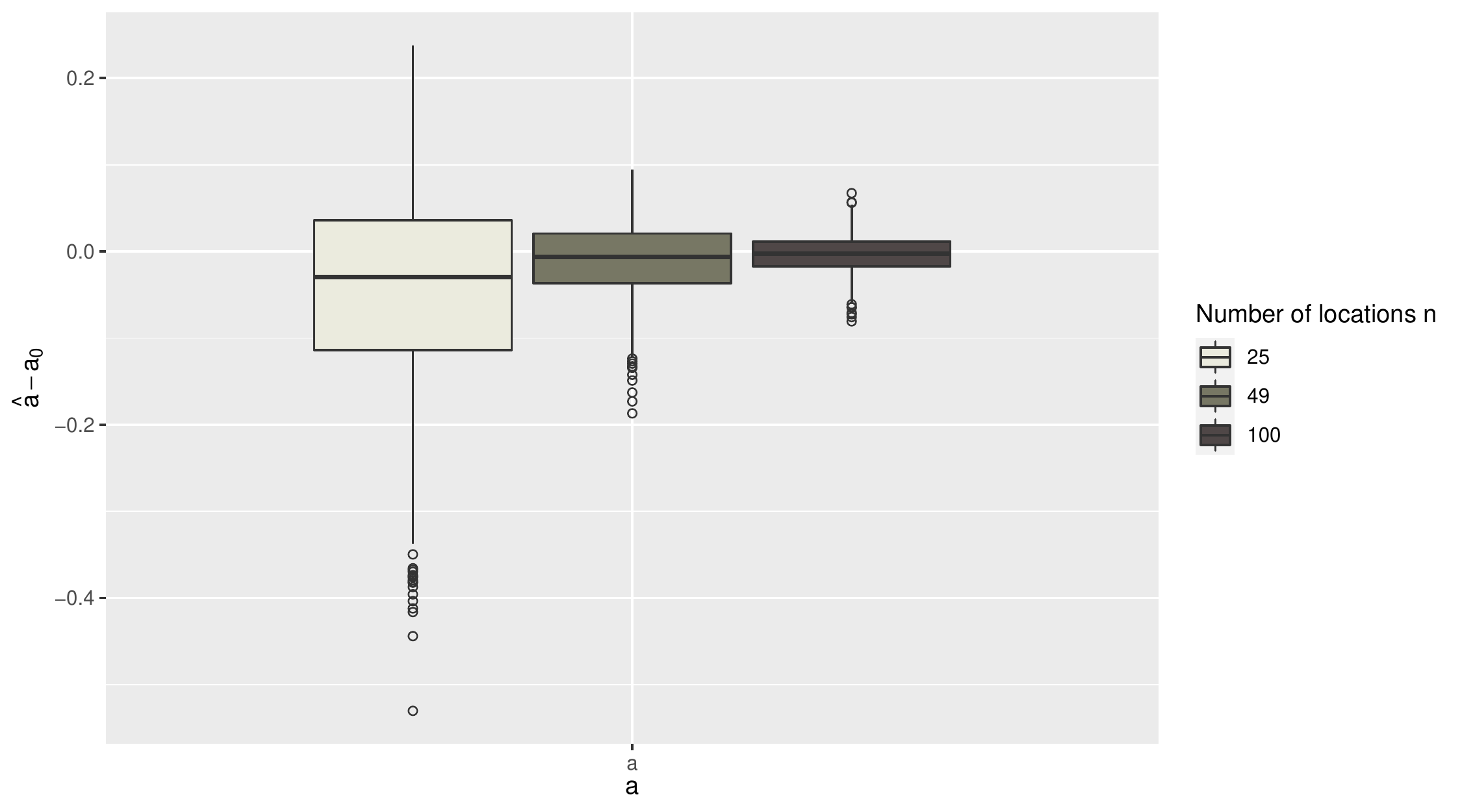}\\
  \includegraphics[width=0.6\textwidth]{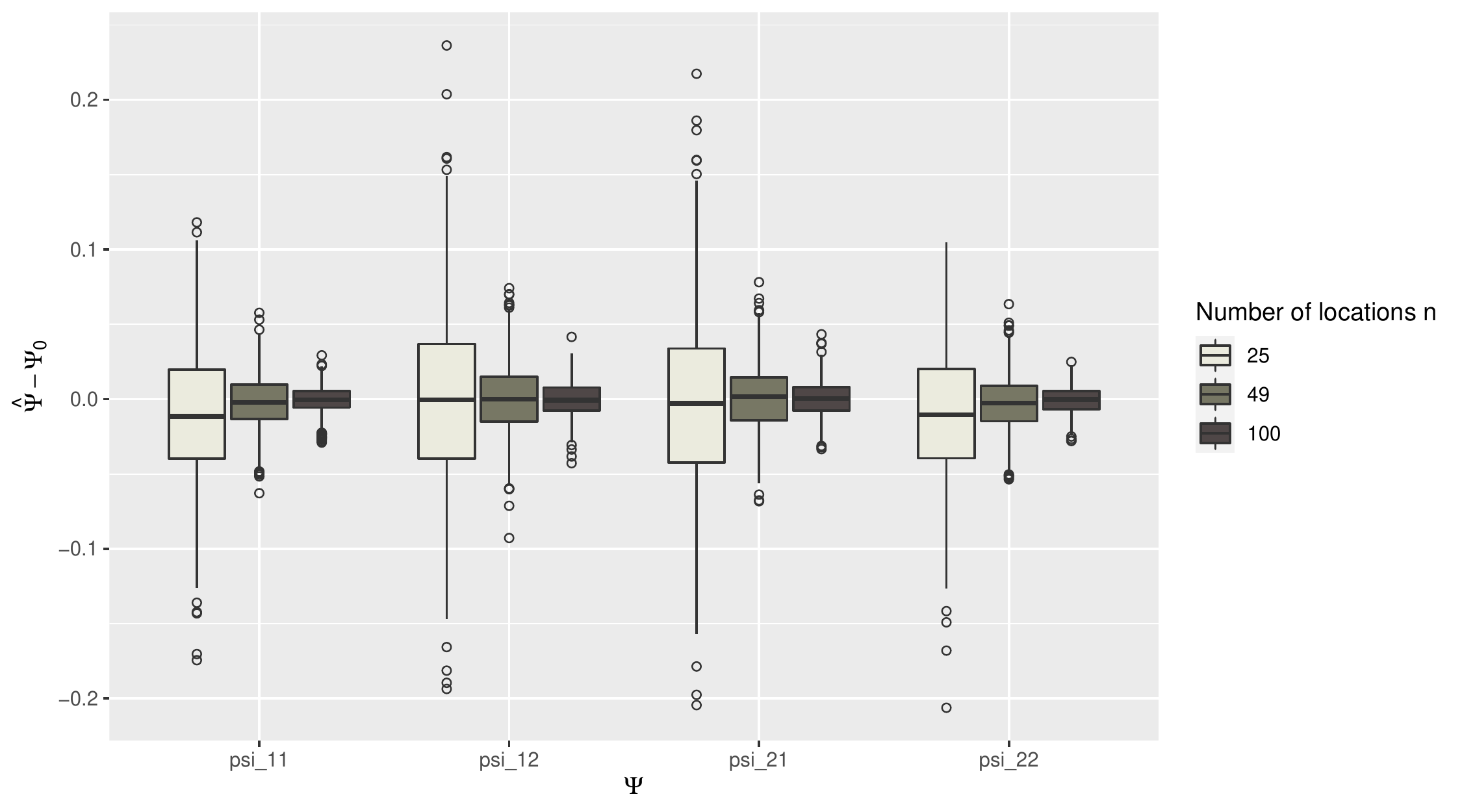}\\
  \includegraphics[width=0.6\textwidth]{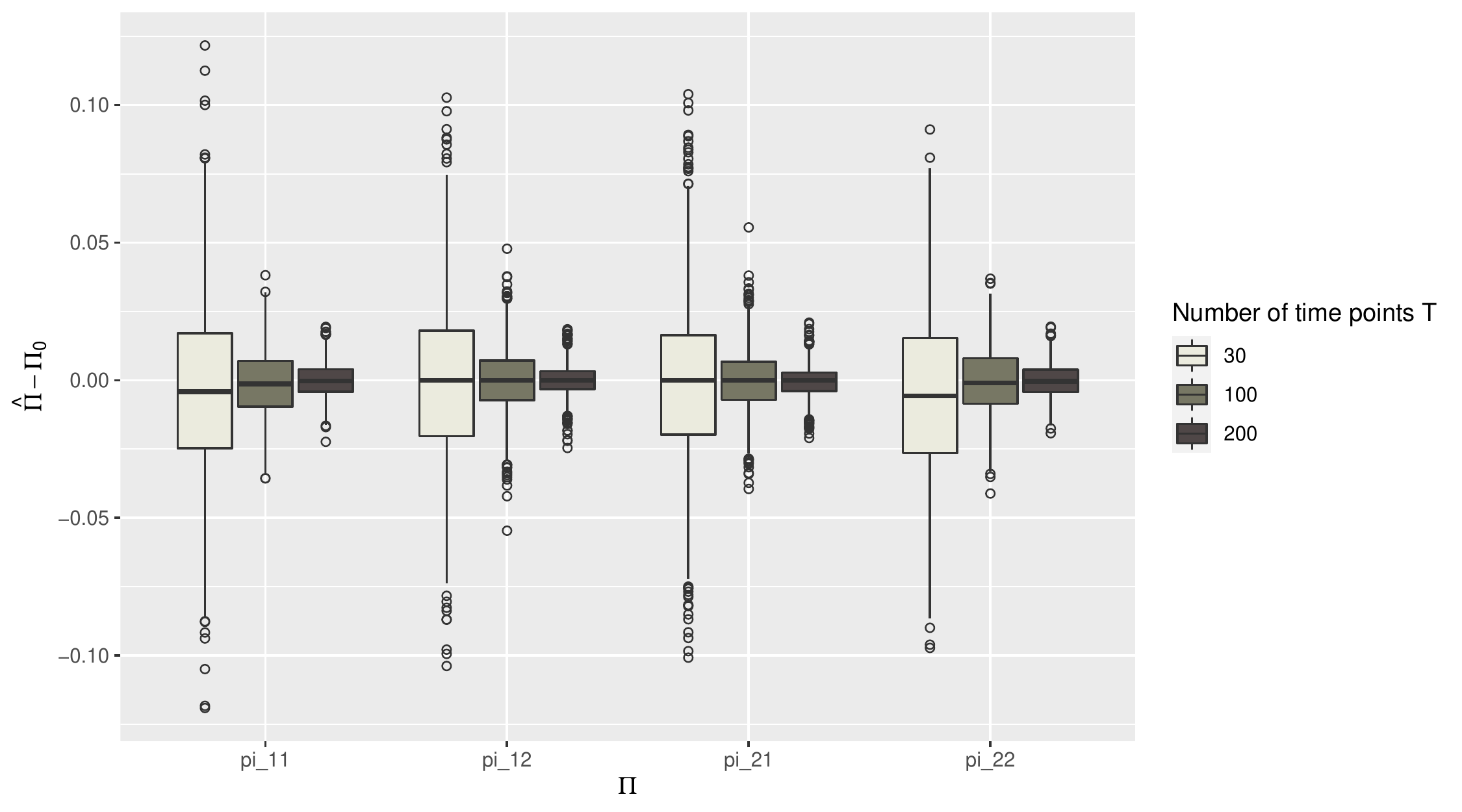}  
  \caption{Estimation performance of the QML estimator for Model A with standard normal errors. Top row: unconditional variance level $a$, centre: spatial coefficient matrix $\xmat{\Psi}$, bottom: temporal coefficient matrix $\xmat{\Pi}$. For each plot, we show the difference between the parameter estimate and the true data-generating parameter.}\label{fig:MC_results}
\end{figure}

\begin{table}
\caption{Average bias of the QML estimates for all considered settings.}\label{table:MC_bias}
\begin{scriptsize}
   \begin{tabular}{l l l ccccccccc}
   \hline
      \multicolumn{3}{c}{Average bias} & $a$ & $\psi_{11}$ & $\psi_{21}$ & $\psi_{12}$ & $\psi_{22}$ & $\pi_{11}$ & $\pi_{21}$ & $\pi_{12}$ & $\pi_{22}$ \\
   \hline
   \multicolumn{3}{c}{Model A, \emph{data-generating parameters}} & 1 & 0.5 & 0.1 & 0.1 & 0.5 & 0.3 & 0 & 0 & 0.3 \\
   &                     & $n = 25$,  $T = 30$   & -0.0464 & -0.0115 & -0.0021 & -0.0003 & -0.0112 & -0.0036 & -0.0010 & -0.0008 & -0.0055 \\
   &       Normal errors & $n = 49$,  $T = 100$  & -0.0093 & -0.0024 &  0.0005 &  0.0002 & -0.0028 & -0.0013 & -0.0002 & -0.0002 & -0.0004 \\
   &                     & $n = 100$, $T = 200$  & -0.0027 & -0.0005 &  0.0001 & -0.0001 & -0.0006 & -0.0001 & -0.0003 &  0.0000 & -0.0003 \\
   &                     & $n = 25$,  $T = 30$   &  0.0269 & -0.0106 & -0.0026 & -0.0040 & -0.0103 & -0.0042 &  0.0006 & -0.0035 & -0.0050 \\
   & $t_3$-distr. errors & $n = 49$,  $T = 100$  &  0.0029 & -0.0005 & -0.0007 & -0.0005 & -0.0010 & -0.0016 & -0.0009 & -0.0002 & -0.0010 \\
   &                     & $n = 100$, $T = 200$  &  0.0009 & -0.0005 &  0.0004 & -0.0006 & -0.0003 & -0.0001 &  0.0001 &  0.0001 &  0.0000 \\ 
   \multicolumn{3}{c}{Model B, \emph{data-generating parameters}} & 1 & 0.5 & 0.1 & 0.1 & 0.5 & 0 & 0 & 0 & 0 \\
   &                     & $n = 25$,  $T = 30$   &  0.0249 &  0.0143 &  0.0194 &  0.0189 &  0.0149 & -0.0014 & -0.0009 & -0.0019 & -0.0004 \\
   &       Normal errors & $n = 49$,  $T = 100$  & -0.0021 & -0.0008 & -0.0005 & -0.0003 & -0.0004 & -0.0001 & -0.0003 &  0.0007 & -0.0004 \\
   &                     & $n = 100$, $T = 200$  & -0.0009 &  0.0001 & -0.0005 &  0.0004 & -0.0005 &  0.0001 &  0.0001 &  0.0000 & -0.0002 \\
   &                     & $n = 25$,  $T = 30$   &  0.0088 &  0.0013 &  0.0006 & -0.0019 & -0.0011 & -0.0026 & -0.0017 & -0.0003 &  0.0006 \\
   & $t_3$-distr. errors & $n = 49$,  $T = 100$  &  0.0017 &  0.0010 &  0.0011 & -0.0008 & -0.0005 & -0.0002 &  0.0000 & -0.0005 &  0.0001 \\
   &                     & $n = 100$, $T = 200$  &  0.0001 &  0.0001 & -0.0011 &  0.0017 & -0.0003 &  0.0000 & -0.0002 &  0.0002 &  0.0002 \\
   \multicolumn{3}{c}{Model C, \emph{data-generating parameters}} & 1 & 0.2 & 0.4 & 0.4 & 0.2 & 0.3 & 0 & 0 & 0.3 \\
   &                     & $n = 25$,  $T = 30$   & -0.0448 & -0.0075 & -0.0057 &  0.0018 & -0.0120 & -0.0037 & -0.0003 & -0.0014 & -0.0031 \\
   &       Normal errors & $n = 49$,  $T = 100$  & -0.0095 & -0.0025 & -0.0003 &  0.0002 & -0.0022 & -0.0002 & -0.0009 & -0.0003 & -0.0008 \\
   &                     & $n = 100$, $T = 200$  & -0.0012 & -0.0004 &  0.0001 &  0.0003 & -0.0001 & -0.0002 & -0.0001 & -0.0002 & -0.0001 \\
   &                     & $n = 25$,  $T = 30$   &  0.0140 & -0.0086 & -0.0048 & -0.0016 & -0.0092 & -0.0029 & -0.0017 & -0.0016 & -0.0033 \\
   & $t_3$-distr. errors & $n = 49$,  $T = 100$  &  0.0023 & -0.0018 & -0.0003 & -0.0006 & -0.0011 &  0.0001 & -0.0002 & -0.0003 & -0.0009 \\
   &                     & $n = 100$, $T = 200$  &  0.0001 &  0.0001 & -0.0011 &  0.0005 & -0.0009 &  0.0000 & -0.0001 &  0.0000 & -0.0003 \\
   \hline	
   \end{tabular}	
\end{scriptsize}
\end{table}

\begin{table}
\caption{Root-mean-square errors of the QML estimates for all considered settings.}\label{table:MC_rmse}
\begin{scriptsize}
   \begin{tabular}{l l l ccccccccc}
   \hline
      \multicolumn{3}{c}{RMSE} & $a$ & $\psi_{11}$ & $\psi_{21}$ & $\psi_{12}$ & $\psi_{22}$ & $\pi_{11}$ & $\pi_{21}$ & $\pi_{12}$ & $\pi_{22}$ \\
   \hline
   \multicolumn{3}{c}{Model A, \emph{data-generating parameters}} & 1 & 0.5 & 0.1 & 0.1 & 0.5 & 0.3 & 0 & 0 & 0.3 \\
   &                     & $n = 25$,  $T = 30$   & 3.9133 & 1.4435 & 1.8352 & 1.8629 & 1.4039 & 1.0310 & 1.0018 & 0.9909 & 0.9897 \\
   &       Normal errors & $n = 49$,  $T = 100$  & 1.3972 & 0.5464 & 0.6884 & 0.7194 & 0.5580 & 0.3827 & 0.3774 & 0.3920 & 0.3846 \\
   &                     & $n = 100$, $T = 200$  & 0.6991 & 0.2737 & 0.3704 & 0.3482 & 0.2810 & 0.1944 & 0.1913 & 0.1899 & 0.1893 \\
   &                     & $n = 25$,  $T = 30$   & 2.9966 & 1.4007 & 2.2693 & 2.2523 & 1.3680 & 0.9902 & 0.9227 & 1.0048 & 1.0624 \\
   & $t_3$-distr. errors & $n = 49$,  $T = 100$  & 0.9105 & 0.5267 & 0.8483 & 0.8098 & 0.5430 & 0.4066 & 0.3638 & 0.3685 & 0.3900 \\
   &                     & $n = 100$, $T = 200$  & 0.4691 & 0.2706 & 0.4123 & 0.4119 & 0.2773 & 0.1856 & 0.1656 & 0.1593 & 0.1875 \\
   \multicolumn{3}{c}{Model B, \emph{data-generating parameters}} & 1 & 0.5 & 0.1 & 0.1 & 0.5 & 0 & 0 & 0 & 0 \\
   &                     & $n = 25$,  $T = 30$   & 4.3186 & 2.8828 & 4.1247 & 4.0821 & 2.9765 & 1.0613 & 0.9657 & 0.9314 & 1.0718 \\
   &       Normal errors & $n = 49$,  $T = 100$  & 0.7958 & 0.6037 & 1.1766 & 1.1640 & 0.5773 & 0.4061 & 0.3906 & 0.3873 & 0.4101 \\
   &                     & $n = 100$, $T = 200$  & 0.3801 & 0.3000 & 0.6349 & 0.6245 & 0.2896 & 0.2050 & 0.1844 & 0.1891 & 0.2087 \\
   &                     & $n = 25$,  $T = 30$   & 2.2659 & 1.6171 & 4.1376 & 4.0416 & 1.5470 & 1.0687 & 0.9542 & 0.9499 & 1.1071 \\
   & $t_3$-distr. errors & $n = 49$,  $T = 100$  & 0.8402 & 0.5914 & 1.6691 & 1.7073 & 0.5883 & 0.4076 & 0.3529 & 0.3483 & 0.4107 \\
   &                     & $n = 100$, $T = 200$  & 0.4101 & 0.3141 & 0.9369 & 0.9369 & 0.3000 & 0.2017 & 0.1665 & 0.1698 & 0.2047 \\
   \multicolumn{3}{c}{Model C, \emph{data-generating parameters}} & 1 & 0.2 & 0.4 & 0.4 & 0.2 & 0.3 & 0 & 0 & 0.3 \\
   &                     & $n = 25$,  $T = 30$   & 3.8513 & 1.8904 & 2.5716 & 2.5115 & 1.8352 & 1.1072 & 0.9469 & 0.9642 & 1.0788 \\
   &       Normal errors & $n = 49$,  $T = 100$  & 1.3681 & 0.7447 & 0.9906 & 0.9952 & 0.7433 & 0.4130 & 0.3763 & 0.3789 & 0.4093 \\
   &                     & $n = 100$, $T = 200$  & 0.7024 & 0.3848 & 0.4801 & 0.5083 & 0.3694 & 0.2018 & 0.1870 & 0.1875 & 0.2073 \\
   &                     & $n = 25$,  $T = 30$   & 2.8111 & 1.8837 & 2.5893 & 2.6870 & 1.8888 & 1.0676 & 0.9646 & 0.9569 & 1.0201 \\
   & $t_3$-distr. errors & $n = 49$,  $T = 100$  & 0.9835 & 0.7359 & 1.0165 & 1.0550 & 0.7755 & 0.4112 & 0.3547 & 0.3723 & 0.3998 \\
   &                     & $n = 100$, $T = 200$  & 0.4660 & 0.3717 & 0.5062 & 0.5149 & 0.3806 & 0.1978 & 0.1754 & 0.1734 & 0.1968 \\
   \hline	
   \end{tabular}	
\end{scriptsize}
\end{table}

\section{Real-World Example: Berlin Real-Estate Prices}

In the following section, we will show the application of the process to a real example. For this purpose, we model the changes in the average sales prices of undeveloped land, developed land and condominiums in Berlin. The data are average monthly average prices per square metre of land or living space in each post-code region from 2002 to 2014. The average prices across all spatial locations are depicted in Figure \ref{fig:data} as time series process.  

\begin{figure}
  \centering
  \includegraphics[width=0.8\textwidth]{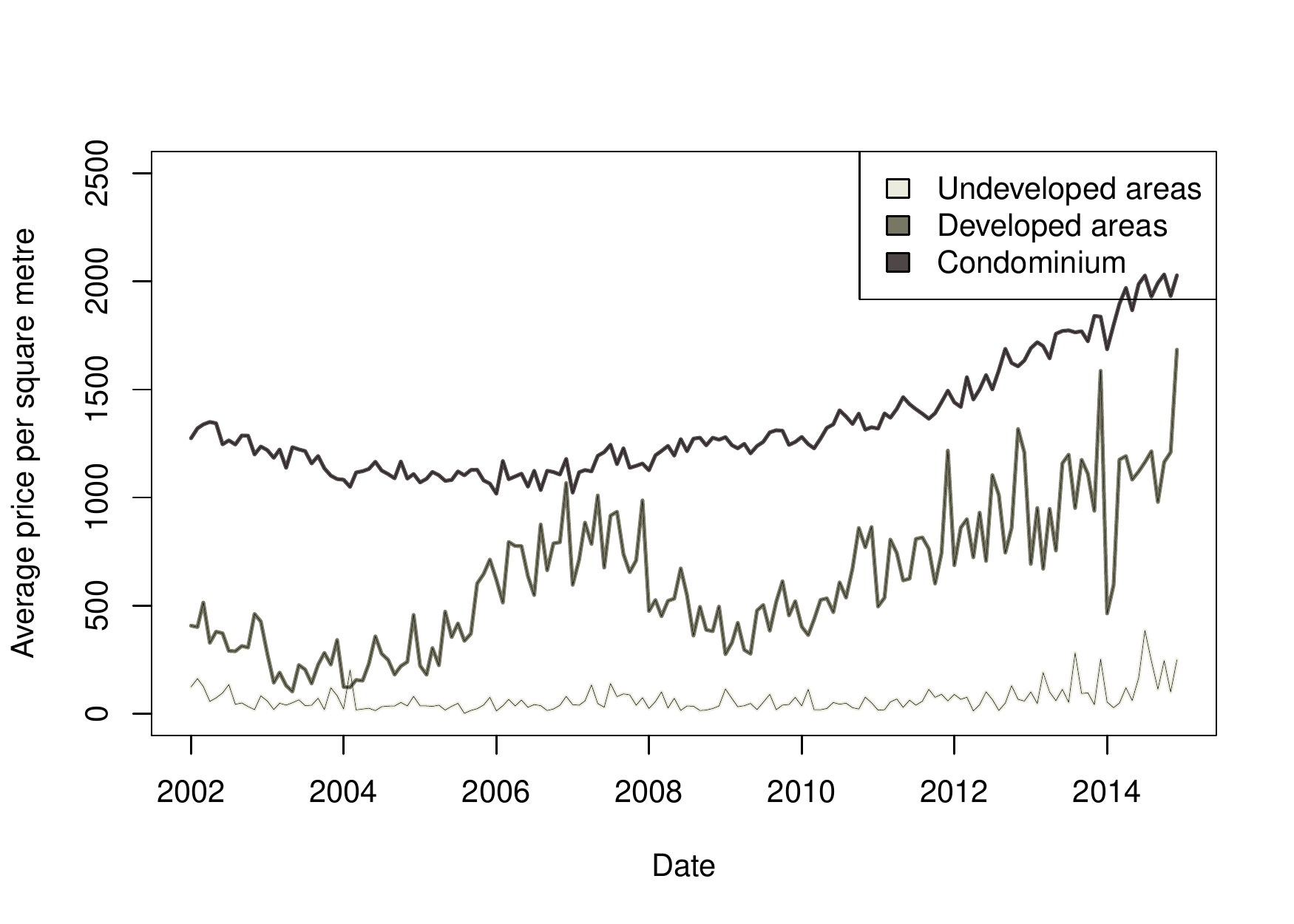}
  \caption{Monthly average prices (Euro/$m^2$) of undeveloped, developed land, and condominium across 190 post-code regions in Berlin from January 2002 to December 2014.}\label{fig:data}
\end{figure}

First of all, it must be noted that there are typically geographical dependencies in the housing market, unlike for other financial markets where trading can take place regardless of location. One of the most important factors in a purchase decision is the location of the property, whereby prices are also influenced by the surrounding neighbourhood. This dependency is in turn influenced by road connections, infrastructure or public transport. Furthermore, the price in the past plays also a role, as is typical for all time series. The temporal proximity creates a causal statistical dependence that decreases the further one looks into the past. These dependencies are observed both in the price process and in the risks in terms of price changes.

This motivates the application of the proposed multivariate spatiotemporal ARCH process to property sales returns. More precisely, we analysed the logarithmic, monthly returns of the average sales prices in each category for all $n = 190$ postcode regions in Berlin. The length of the time series is accordingly $T = 156$ and the process is $p = 3$-dimensional. To display the log-return process, Figure \ref{fig:returns} shows the average log-returns across all locations in the temporal domain. Especially for the developed and undeveloped land, there were much fewer sales, such that the average returns are more volatile. In the case of no transactions in certain months and areas, we assumed that the average sales price did not change and, thus, the log-returns are zero. More precisely, we randomly simulated a normally distributed return with mean zero and standard deviation 0.0001 to not have positive probability for zero returns. In future, a more detailed analysis including a zero-transaction model would be interesting, especially for smaller time granularities and spatial locations.

\begin{figure}
  \centering
  \includegraphics[width=0.6\textwidth]{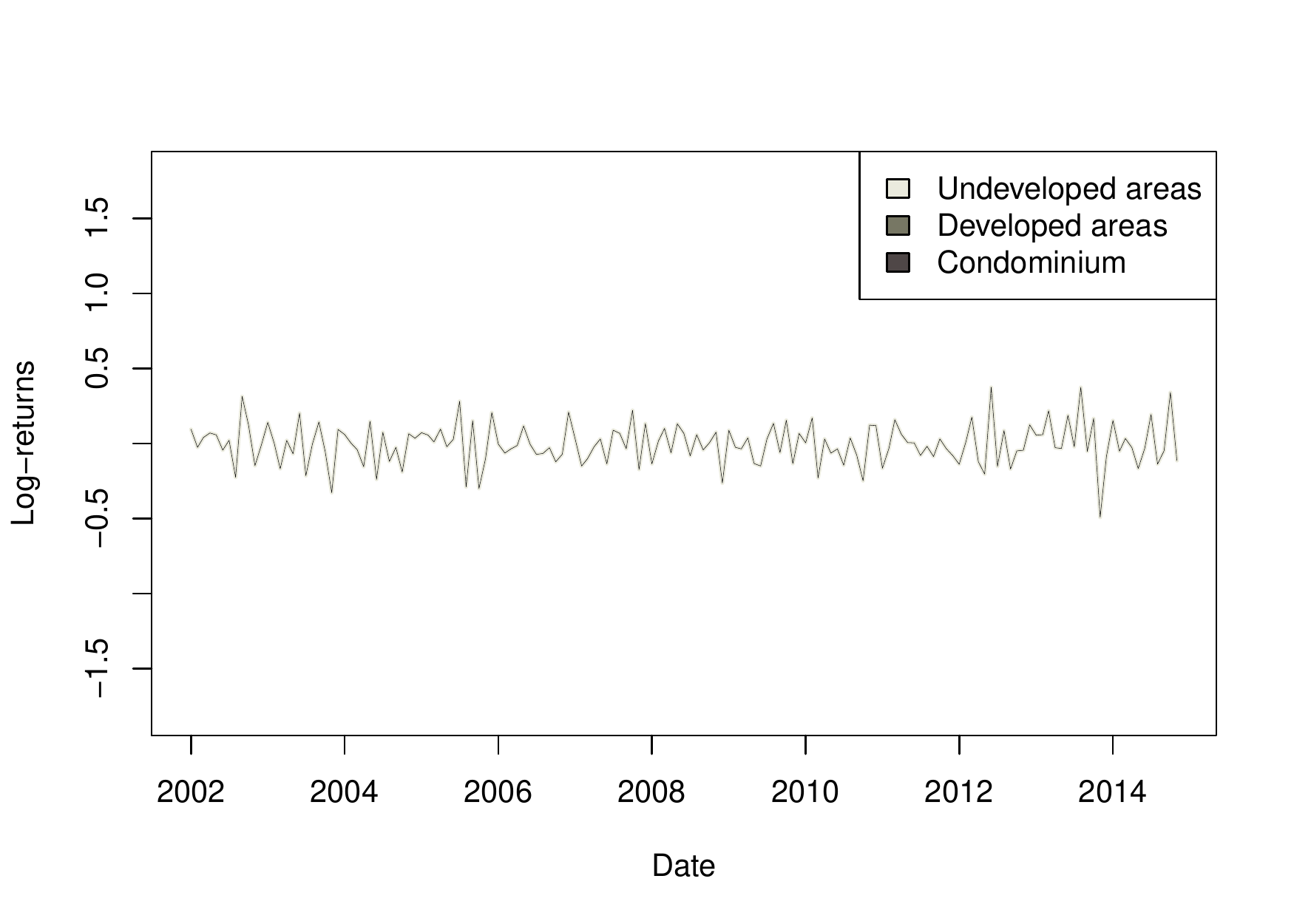}\\
  \includegraphics[width=0.6\textwidth]{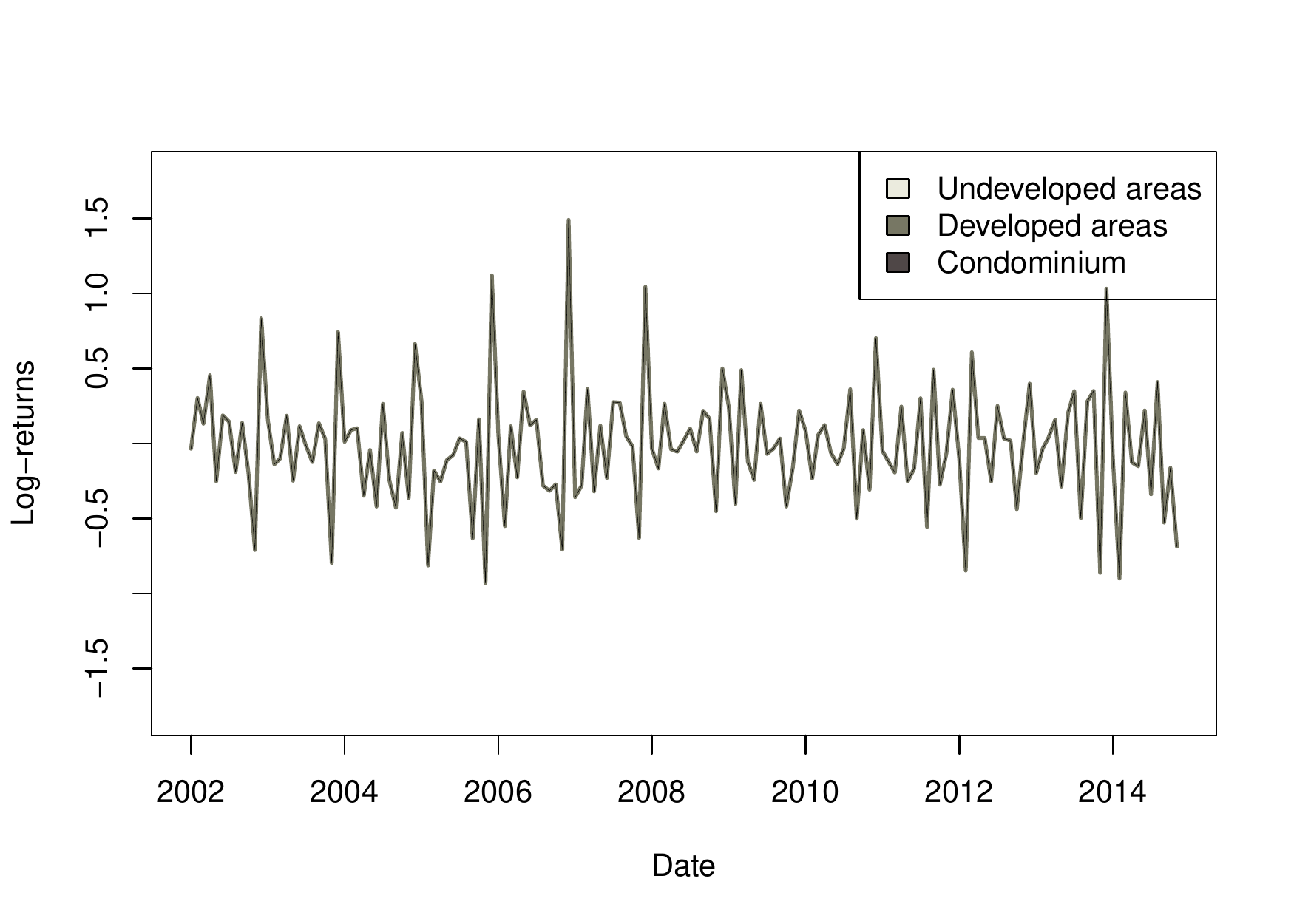}\\
  \includegraphics[width=0.6\textwidth]{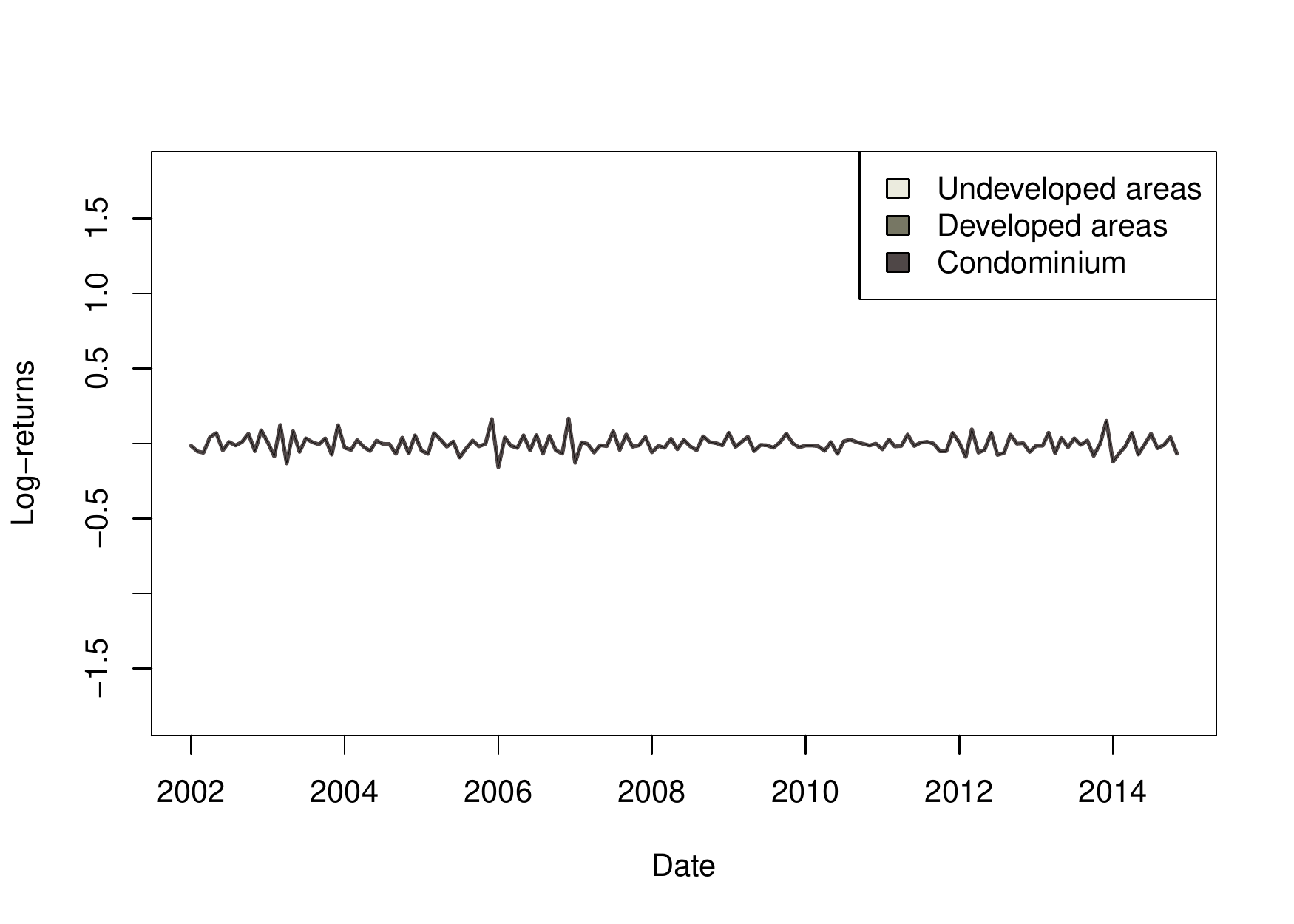}\\
  \caption{Monthly average log-returns of the price series displayed in Figure \ref{fig:data}. Note that the log-returns were computed from the raw spatiotemporal data and averaged afterwards to be depicted as time series.}\label{fig:returns}
\end{figure}

The estimated parameters of the multivariate spatiotemporal ARCH process are reported in Table \ref{table:qmlexample} along with their standard errors. The unconditional variance levels were assumed to be constant across space, but vary with the property types. Bearing in mind that we have modelled monthly returns, we observe interesting results. First, the spatial dependence is dominated by the temporal dependence that appears to be more important. Second, spatial spill-overs are positive (i.e., we observed clusters of higher variances/risks), but they are only significant for developed land. When increasing the temporal intervals from monthly to quarterly data, these spatial interactions will disappear. The same holds when grouping the spatial locations to larger areas. This highlights the importance of spatial GARCH models for small spatial units and time granularities (as it is also well-known in finance). Third, cross-variable spill-overs are only significant at the first temporal lag (i.e., after one time period). More precisely, we see significant interactions only between developed and undeveloped land, but not for condominium prices. It is important to bear in mind that the spatial and temporal ARCH effects will also cover changes in the variance due to latent variables. Fourth, the unconditional variance varies across the property types with developed land experiencing the highest variance, followed by the condominium and undeveloped land. Note that undeveloped land usually does not have and will not get building permission.

% 0.02897228 4.06269935 0.25639217 % exp(est$A_tilde_est[1,] - mean(log(rnorm(1000000)^2)))

\begin{table}
\caption{QML estimates and standard errors of the empirical example. Spatial ARCH effects are highlighted in light green, temporal ARCH effects in dark green. Significant effects are marked by an asterisk (* $t$-value $>$ 1.9, ** $t$-value $>$ 2).}\label{table:qmlexample}
\begin{center}
\begin{scriptsize}
	\begin{tabular}{c l cc cc cc}
	\hline 
	& & \multicolumn{2}{c}{Undeveloped land} & \multicolumn{2}{c}{Developed land}   & \multicolumn{2}{c}{Condominium}      \\
	& & Estimate & Standard error            & Estimate & Standard error            & Estimate & Standard error            \\
	\hline
	\multicolumn{2}{l}{$\tilde{\xmat{A}}$} & -4.686** & 1.381 & 0.187 & 1.372 & -2.652* & 1.337 \\
	\hline
	              & Undeveloped land       & \cellcolor{col1} 0.111 & \cellcolor{col1} 0.074 &                  0.016   &                  0.075 &                 -0.057 &                  0.074 \\
	$\xmat{\Psi}$ & Developed land         &                  0.014 &                  0.064 & \cellcolor{col1} 0.144** & \cellcolor{col1} 0.062 &                  0.000 &                  0.064 \\
	              & Condominium            &                 -0.085 &                  0.090 &                  0.008   &                  0.090 & \cellcolor{col1} 0.113 & \cellcolor{col1} 0.086 \\
	\hline
	              & Undeveloped land       & \cellcolor{col2} 0.583** & \cellcolor{col2} 0.038 &                  0.129** &                  0.038 &                  -0.014   &                  0.038 \\
	$\xmat{\Pi}$  & Developed land         &                  0.080** &                  0.031 & \cellcolor{col2} 0.553** & \cellcolor{col2} 0.031 &                   0.027   &                  0.031 \\
	              & Condominium            &                 -0.028   &                  0.044 &                  0.078   &                  0.044 & \cellcolor{col2}  0.606** & \cellcolor{col2} 0.044 \\

	\hline
	\end{tabular}
\end{scriptsize}
\end{center}
\end{table}

\section{Summary and Conclusion}\label{sec:conclusion}

In this paper, we have introduced a multivariate spatiotemporal autoregressive model for the conditional heteroscedasticity (multivariate vec-spARCH). While ARCH and GARCH models are well-known in time-series econometrics and finance, spatiotemporal extensions typically did not account for spatial simultaneity. That is, for any geographical phenomena, spatial interactions occur instantaneous due to the spatial proximity between the observations. Instead, previous papers typically only allowed for time-lagged spatial dependence. The model introduced in this paper explicitly accounts for instantaneous spatial and cross-variable interactions and temporal dependence in the conditional variance. Thus, the model would also be suitable to model purely spatial data without the need of observations over time. In the empirical application, it gets obvious that the log-returns of several types of real estate are spatially autocorrelated. This indicate local clusters of increased volatilities and market risks -- even though the temporal dependence appears to be more important. Thus, we could show that there are temporally and spatially varying volatilities. Furthermore, we found significant cross-variable dependence in the first temporal lags, but no significant instantaneous cross-variable interactions. This again motivates the application of a multivariate spatiotemporal ARCH model in such studies.

For this new model, we discussed the parameter estimation using a quasi-maximum-likelihood (QML) approach. For this reason, the process is reformulated in a vec-representation and a log-squared transformation is applied to obtain multivariate spatiotemporal autoregressive process. We showed the consistency of the QML estimator under regular assumptions for the error process when the spatial and temporal dimensions increase. In the finite-sample case, we could see rapidly decreasing root-mean-square errors (RMSEs) in a series of simulations with different model specifications and error distributions. All our simulations could be performed in a reasonable amount of time using a standard computer. The required computational resources are usually the bottleneck of the QML approach due to the computation of the log-determinant of the Jacobian matrix.

There are many further directions for future research and potential fields of applications. First, we did only considered logarithmic structures in the volatility models, but no classical ARCH structures. However, since the multivariate spatiotemporal could be transformed to purely spatial models using the vec-representation, previous results of spatial ARCH and GARCH models could be applied. Furthermore, all these spatial econometric models rely on a (correctly) specified spatial weight matrix, which is, however, mostly unknown in practice. Hence, estimation methods for the entire spatial dependence structures (i.e., each spatial weight) are desirable from a practical perspective. Penalised methods seem to be promising in this case, because many links can be considered to be zero. 

Apart from applications in econometrics, also environmental and climate processes would be interesting and potential fields for application of the multivariate ARCH model. The process parameters can be interpreted as local risk measures, which is highly relevant in environmental studies. Furthermore, when considering our model as error process, spatially and temporally varying measurement or modelling uncertainties can be reflected in the statistical model. For example, this might be of interest for GNSS positioning in urban environments. Other fields, where local risks and cross-variable interactions are highly relevant, are epidemiological and medical studies.

\section*{Appendix}

\begin{appendix}

\section{Proofs}

\begin{proof}[Proof of Proposition \ref{prop:stability}]	
If all eigenvalues of $\xmat{S}_{np}^{-1}(\xmat{\Pi}' \otimes \xmat{I})$ are smaller than one and $(\xmat{I} - \xmat{S}_{np}^{-1}(\xmat{\Pi}' \otimes \xmat{I}))^{j} \rightarrow 0$ for an increasing power $j$, we get that 
\begin{equation}
	(\xmat{I} + \xmat{S}_{np}^{-1}(\xmat{\Pi}' \otimes \xmat{I}) + \ldots + (\xmat{S}_{np}^{-1}(\xmat{\Pi}' \otimes \xmat{I}))^j) vec(\tilde{\xmat{A}}) \rightarrow (\xmat{I} - \xmat{S}_{np}^{-1}(\xmat{\Pi}' \otimes \xmat{I}))^{-1} \, 
\end{equation}
and
\begin{equation}
	\ddot{\xvec{Y}}_t  =  (\xmat{I} - \xmat{S}_{np}^{-1}(\xmat{\Pi}' \otimes \xmat{I}))^{-1} + \sum_{i = 0}^{\infty} (\xmat{S}_{np}^{-1}(\xmat{\Pi}' \otimes \xmat{I}))^i \xmat{U}_{t-i} \, .
\end{equation}
The stability follows from the convergence of the $(\xmat{S}_{np}^{-1}(\xmat{\Pi}' \otimes \xmat{I}))^i$. If the spectral radius of $\xmat{S}_{np}^{-1}(\xmat{\Pi}' \otimes \xmat{I})$ is smaller than one, $\xmat{S}_{np}^{-1}(\xmat{\Pi}' \otimes \xmat{I}) \rightarrow 0$ (e.g., \citealt{gentle2017matrix}).
\end{proof}

\begin{proof}[Proof of Proposition \ref{prop:identification}]	
We have to show that	
\begin{equation*}
	\frac{1}{Tnp} E\left(\ln \mathcal{L}( \tilde{\xmat{A}}, \xmat{\Psi}, \xmat{\Pi} | \xmat{Y}_0)\right)  -  \frac{1}{Tnp} E\left(\ln \mathcal{L}( \tilde{\xmat{A}}_0, \xmat{\Psi}_0, \xmat{\Pi}_0 | \xmat{Y}_0)\right)  \leq 0, \label{eq:ident}
\end{equation*}
where the equality holds if and only if $\tilde{\xmat{A}} = \tilde{\xmat{A}}_0$, $\xmat{\Psi} = \xmat{\Psi}_0$, and $\xmat{\Pi} =  \xmat{\Pi}_0$. 
\begin{eqnarray*}
&& \frac{1}{Tnp} E\left(\ln \mathcal{L}( \tilde{\xmat{A}}, \xmat{\Psi}, \xmat{\Pi} | \xmat{Y}_0)\right) -  \frac{1}{Tnp} E\left(\ln \mathcal{L}( \tilde{\xmat{A}}_0, \xmat{\Psi}_0, \xmat{\Pi}_0 | \xmat{Y}_0)\right)\\
	& =  & \frac{1}{np} (\ln |\xmat{S}_{np}| - \ln |\xmat{S}_{np0}|) \\ 
                              &   & - \quad \frac{1}{2np} \sum_{t = 1}^{T} \left[ \xmat{S}_{np} \xmat{S}_{np0}^{-1} (vec(\tilde{\xmat{A}}_0 - \tilde{\xmat{A}}) + ((\xmat{\Pi}_0' - \xmat{\Pi}') \otimes \xmat{I}) \ddot{\xvec{Y}}_{t-1}) \right]' \\
                              &   & \qquad\qquad\qquad \times \left[ \xmat{S}_{np} \xmat{S}_{np0}^{-1} (vec(\tilde{\xmat{A}}_0 - \tilde{\xmat{A}}) + ((\xmat{\Pi}_0' - \xmat{\Pi}') \otimes \xmat{I}) \ddot{\xvec{Y}}_{t-1}) \right]\, \\
                              &   & - \quad \frac{1}{2np} tr\left( \xmat{S}_{np} \xmat{S}_{np0}^{-1} \xmat{S}_{np0}^{'-1} \xmat{S}_{np}'\right)\\
    & =  & \frac{1}{np} \ln |\xmat{S}_{np}\xmat{S}_{np0}^{-1}\xmat{S}_{np0}^{'-1}\xmat{S}_{np}'|^{1/np} \\ 
                              &   & - \quad \frac{1}{2np} \sum_{t = 1}^{T} \left[ \xmat{S}_{np} \xmat{S}_{np0}^{-1} (vec(\tilde{\xmat{A}}_0 - \tilde{\xmat{A}}) + ((\xmat{\Pi}_0' - \xmat{\Pi}') \otimes \xmat{I}) \ddot{\xvec{Y}}_{t-1}) \right]' \\
                              &   & \qquad\qquad\qquad \times \left[ \xmat{S}_{np} \xmat{S}_{np0}^{-1} (vec(\tilde{\xmat{A}}_0 - \tilde{\xmat{A}}) + ((\xmat{\Pi}_0' - \xmat{\Pi}') \otimes \xmat{I}) \ddot{\xvec{Y}}_{t-1}) \right]\, \\
                              &   & - \quad \frac{1}{2np} tr\left( \xmat{S}_{np} \xmat{S}_{np0}^{-1} \xmat{S}_{np0}^{'-1} \xmat{S}_{np}'\right)
\end{eqnarray*}

First,  we focus on the convergence of the quadratic term
\begin{equation}
	 \frac{1}{2np} \sum_{t = 1}^{T}  \xmat{V}_t' \xmat{V}_t \label{eq:vt}
\end{equation}
with
\begin{eqnarray*}
 \xmat{V}_t	& = & \xmat{S}_{np} \xmat{S}_{np0}^{-1} (vec(\tilde{\xmat{A}}_0 - \tilde{\xmat{A}}) + ((\xmat{\Pi}_0' - \xmat{\Pi}') \otimes \xmat{I}) \ddot{\xvec{Y}}_{t-1})\\
	& = & \left( \xmat{I} - ((\xmat{\Psi}' - \xmat{\Psi}_0') \otimes \xmat{W}) \xmat{S}_{np0}^{-1} \right) (vec(\tilde{\xmat{A}}_0 - \tilde{\xmat{A}}) + ((\xmat{\Pi}_0' - \xmat{\Pi}') \otimes \xmat{I}) \ddot{\xvec{Y}}_{t-1}) \, .
	%& = & vec(\tilde{\xmat{A}}_0 - \tilde{\xmat{A}}) + ((\xmat{\Pi}_0' - \xmat{\Pi}') \otimes \xmat{I}) \ddot{\xvec{Y}}_{t-1} \\
	%&& - ((\xmat{\Psi}' - \xmat{\Psi}_0') \otimes \xmat{W}) \xmat{S}_{np0}^{-1} vec(\tilde{\xmat{A}}_0 - \tilde{\xmat{A}}) - ((\xmat{\Psi}' - \xmat{\Psi}_0') \otimes \xmat{W}) \xmat{S}_{np0}^{-1} ((\xmat{\Pi}_0' - \xmat{\Pi}') \otimes \xmat{I}) \ddot{\xvec{Y}}_{t-1}
\end{eqnarray*}
Thus, under Assumption \ref{ass:5}, \eqref{eq:vt} is equal to zero if and only if $\tilde{\xmat{A}} = \tilde{\xmat{A}}_0$ and $\xmat{\Pi} = \xmat{\Pi}_0$. Note that $\tilde{\xmat{A}}_0$ is constant across time, while $\ddot{\xvec{Y}}_t$ is varying due to the random variation in $\xmat{\Xi}_t$. Thus, if $T = 1$, $\tilde{\xmat{A}}_0$ must be assumed to be constant across space, i.e., $\tilde{\xmat{A}}_0 = \tilde{a}_0 \xvec{1}_n$, to obtain identifiability.

% Since $\xmat{S}_{np}(\xmat{\Psi})\xmat{S}_{np0}^{-1} = \xmat{I} - ((\xmat{\Psi}' - \xmat{\Psi}_0') \otimes \xmat{W}) \xmat{S}_{np0}^{-1} \geq 0$ \red{really?}, we get that \eqref{eq:vt} is zero
Second,
\begin{equation*}
	\frac{1}{2np} tr\left( \xmat{S}_{np} \xmat{S}_{np0}^{-1} \xmat{S}_{np0}^{'-1} \xmat{S}_{np}'\right)	% \frac{1}{np} \ln |\xmat{S}_{np}\xmat{S}_{np0}^{-1}\xmat{S}_{np0}^{'-1}\xmat{S}_{np}'|^{1/np} - \frac{1}{2np} tr\left( \xmat{S}_{np} \xmat{S}_{np0}^{-1} \xmat{S}_{np0}^{'-1} \xmat{S}_{np}'\right)
\end{equation*}
is only a function of $\xmat{\Psi}$ and 
\begin{equation}
	\frac{1}{np} tr\left( \xmat{S}_{np} \xmat{S}_{np0}^{-1} \xmat{S}_{np0}^{'-1} \xmat{S}_{np}'\right)	 \geq |\xmat{S}_{np}\xmat{S}_{np0}^{-1}\xmat{S}_{np0}^{'-1}\xmat{S}_{np}'|^{1/np}.  \label{eq:amgm}
\end{equation}
by the arithmetic and geometric means inequality of eigenvalues of $\xmat{S}_{np}$. Further,
\begin{equation*}
	\xmat{S}_{np}(\xmat{\Psi})\xmat{S}_{np0}^{-1} = \xmat{I} - ((\xmat{\Psi}' - \xmat{\Psi}_0') \otimes \xmat{W}) \xmat{S}_{np0}^{-1} 
\end{equation*}
is equal to $\xmat{I}$ if and only if $\xmat{\Psi} = \xmat{\Psi}_0$ as $n \to \infty$. Then, the equality of \eqref{eq:amgm} holds. 

As consequence, \eqref{eq:ident} is equal to zero if and only if the parameters coincide with their true values $\xmat{\Psi}_0$, $\xmat{\Pi}_0$, and $\tilde{\xmat{A}}_0$. Hence, the parameters are uniquely identifiable.

% $\ln |\xmat{S}_{np}| - \ln |\xmat{S}_{np0}| = \ln |\xmat{S}_{np}\xmat{S}_{np0}^{-1}\xmat{S}_{np0}^{'-1}\xmat{S}_{np}'|^{1/p}$ ...
%With $\xmat{S}_{np}(\xmat{\Psi})\xmat{S}_{np0}^{-1} = \xmat{I} - ((\xmat{\Psi}' - \xmat{\Psi}_0') \otimes \xmat{W}) \xmat{S}_{np0}^{-1}$, we get
\end{proof}

\begin{lemma}[\cite{yang2017identification}, Lemma 1]\label{lemma:UB_Psi}
The sequences $\xmat{S}_{np}$ and $\xmat{S}_{np}^{-1}$ are uniformly bounded in column sum norm, uniformly in $\xmat{\Psi}$, if $\sup_{\xmat{\Psi},n} || \xmat{\Psi}' \otimes \xmat{W}_n ||_1 < 1$. They are uniformly bounded in row sum norm, uniformly in $\xmat{\Psi}$, if $\sup_{\xmat{\Psi},n} || \xmat{\Psi}' \otimes \xmat{W}_n ||_\infty < 1$.
\end{lemma}

\begin{lemma}\label{lemma:conv}
Under Assumptions \ref{ass:1}, \ref{ass:4} and \ref{ass:5}, it holds for an $np$-dimensional non-stochastic, uniformly bounded matrix $\xmat{B}_{np}$ that
\begin{small}
\begin{eqnarray}
	\frac{1}{npT} \sum_{t = 1}^{T} \ln vec(\xmat{Y}_t^{(2)})' \xmat{B}_{np} \ln vec(\xmat{Y}_t^{(2)}) - \frac{1}{npT} E\left[ \sum_{t = 1}^{T} \ln vec(\xmat{Y}_t^{(2)})' \xmat{B}_{np} \ln vec(\xmat{Y}_t^{(2)}) \right] = O_p\left(\frac{1}{\sqrt{npT}}\right), \\
	\frac{1}{npT} \sum_{t = 1}^{T} \ln vec(\xmat{Y}_t^{(2)})' \xmat{B}_{np} vec(\xmat{U}_t) - \frac{1}{npT} E\left[ \sum_{t = 1}^{T} \ln vec(\xmat{Y}_t^{(2)})' \xmat{B}_{np} vec(\xmat{U}_t)\right] = O_p\left(\frac{1}{\sqrt{npT}}\right), \\
	\frac{1}{npT} \sum_{t = 1}^{T} vec(\xmat{U}_t)' \xmat{B}_{np} vec(\xmat{U}_t) - \frac{1}{npT} E\left[ \sum_{t = 1}^{T} vec(\xmat{U}_t)' \xmat{B}_{np} vec(\xmat{U}_t)\right] = O_p\left(\frac{1}{\sqrt{npT}}\right) , & 
\end{eqnarray}
\end{small}
where $E\left[ \sum_{t = 1}^{T} \ln vec(\xmat{Y}_t^{(2)})' \xmat{B}_{np} \ln vec(\xmat{Y}_t^{(2)}) \right]$ is $O(1)$, $E\left[ \sum_{t = 1}^{T} \ln vec(\xmat{Y}_t^{(2)})' \xmat{B}_{np} vec(\xmat{U}_t)\right]$ is $O(1/T)$ and $E\left[ \sum_{t = 1}^{T} vec(\xmat{U}_t)' \xmat{B}_{np} vec(\xmat{U}_t)\right]$ is $O(1)$.
\end{lemma}

\begin{proof}[Proof of Theorem \ref{th:consistency}]	
	The proof of the theorem consists of two parts; first, the identification of the parameters, and, second, the uniform and equicontinuous convergence of $\frac{1}{npT}\ln \mathcal{L}_{nT}(\xvec{\vartheta}|\xvec{Y}_0)$ to $\frac{1}{nT} Q(\xvec{\vartheta}|\xvec{Y}_0)$ in probability with $\xvec{\vartheta}_0$ being a unique maximiser of $Q(\xvec{\vartheta})$. Then, the consistency of the QML estimator follows.
	\begin{itemize}
		\item[1.] The unique identification of the parameters is shown in the proof of Proposition \ref{prop:identification}.
		\item[2.] Let 
		\[ \tilde{\xvec{U}}_t(\vartheta) = \xmat{S}_{np} \ln vec(\xmat{Y}_t^{(2)}) - vec(\tilde{\xmat{A}}) - (\xmat{\Pi}' \otimes \xmat{I})  \ln vec(\xmat{Y}_{t-1}^{(2)}) \]
		 and
		\[ \xvec{U}_t = \xmat{S}_{np0} \ln vec(\xmat{Y}_t^{(2)}) - vec(\tilde{\xmat{A}}_0) - (\xmat{\Pi}_0' \otimes \xmat{I})  \ln vec(\xmat{Y}_{t-1}^{(2)}) , \]
		the true error vector of the transformed process, i.e., $\ln vec(\xmat{\Xi}_t^{(2)})$. Furthermore, let $\ddot{\xvec{Y}}_t = \ln vec(\xmat{Y}_t^{(2)})$ and 
		\[ \tilde{\xvec{U}}_t(\xi) = \xvec{U}_t - (\xmat{\Psi}' - \xmat{\Psi}_0')\otimes \xmat{W} \ddot{\xvec{Y}}_t - vec(\tilde{\xmat{A}} - \tilde{\xmat{A}}_0) - ((\xmat{\Pi}' - \xmat{\Pi}_0') \otimes \xmat{I}) \ddot{\xvec{Y}}_{t-1}  \]
		with $\xi$ being the differences in the parameters. Then,
		\begin{eqnarray*}
			\tilde{\xvec{U}}_t(\xi)'\tilde{\xvec{U}}_t(\xi) & = & \xvec{U}_t'\xvec{U}_t - vec(\tilde{\xmat{A}} - \tilde{\xmat{A}}_0)'vec(\tilde{\xmat{A}} - \tilde{\xmat{A}}_0) \\
			& & + \quad  (\xmat{\Psi}' - \xmat{\Psi}_0')' \otimes \xmat{W}' \ddot{\xvec{Y}}_t' \ddot{\xvec{Y}}_t \xmat{W} \otimes (\xmat{\Psi}' - \xmat{\Psi}_0') \\
			& & - \quad  ((\xmat{\Pi}' - \xmat{\Pi}_0') \otimes \xmat{I})' \ddot{\xvec{Y}}_{t-1}' \ddot{\xvec{Y}}_{t-1} ((\xmat{\Pi}' - \xmat{\Pi}_0') \otimes \xmat{I}) \\
			& & + \quad  2 (\xmat{\Psi}' - \xmat{\Psi}_0')' \otimes \xmat{W}' \ddot{\xvec{Y}}_t' \ddot{\xvec{Y}}_{t-1} ((\xmat{\Pi}' - \xmat{\Pi}_0') \otimes \xmat{I}) \\
			& & - \quad  2 (\xmat{\Psi}' - \xmat{\Psi}_0')' \otimes \xmat{W}' \ddot{\xvec{Y}}_t' \xvec{U}_t\\
			& & - \quad  2 ((\xmat{\Pi}' - \xmat{\Pi}_0') \otimes \xmat{I})' \ddot{\xvec{Y}}_{t-1}' \xvec{U}_t\\
			& & + \quad  2 (\xmat{\Psi}' - \xmat{\Psi}_0')'\otimes \xmat{W}' \ddot{\xvec{Y}}_{t}'  vec(\tilde{\xmat{A}} - \tilde{\xmat{A}}_0)\\
			& & + \quad  2 ((\xmat{\Pi}' - \xmat{\Pi}_0') \otimes \xmat{I})' \ddot{\xvec{Y}}_{t-1}'  vec(\tilde{\xmat{A}} - \tilde{\xmat{A}}_0)\\
			& & - \quad  2 vec(\tilde{\xmat{A}} - \tilde{\xmat{A}}_0)'   \xvec{U}_t \, .
		\end{eqnarray*}
		Using Lemmata \ref{lemma:UB_Psi} and \ref{lemma:conv}, it follows that
		\begin{itemize}
			\item[*] $\frac{1}{npT} \sum_{t = 1}^{T} \xvec{U}_t'\xvec{U}_t - \frac{1}{npT} E \left[ \sum_{t = 1}^{T} \xvec{U}_t'\xvec{U}_t\right] \overset{p}{\to} 0, $
			\item[*] $\frac{1}{npT} \sum_{t = 1}^{T} (\xmat{W} \otimes \xmat{I})' \ddot{\xvec{Y}}_t' \ddot{\xvec{Y}}_t(\xmat{W} \otimes \xmat{I}) - \frac{1}{npT} E \left[ \sum_{t = 1}^{T}  (\xmat{W} \otimes \xmat{I})' \ddot{\xvec{Y}}_t' \ddot{\xvec{Y}}_t(\xmat{W} \otimes \xmat{I}) \right] \overset{p}{\to} 0, $
			\item[*] $\frac{1}{npT} \sum_{t = 1}^{T} \ddot{\xvec{Y}}_{t-1}' \ddot{\xvec{Y}}_{t-1} - \frac{1}{npT} E \left[ \sum_{t = 1}^{T}  \ddot{\xvec{Y}}_{t-1}' \ddot{\xvec{Y}}_{t-1} \right] \overset{p}{\to} 0, $
			\item[*] $\frac{1}{npT} \sum_{t = 1}^{T} (\xmat{W} \otimes \xmat{I})' \ddot{\xvec{Y}}_t'  \ddot{\xvec{Y}}_{t-1} - \frac{1}{npT} E \left[ \sum_{t = 1}^{T}  (\xmat{W} \otimes \xmat{I})' \ddot{\xvec{Y}}_t'  \ddot{\xvec{Y}}_{t-1} \right] \overset{p}{\to} 0, $
			\item[*] $\frac{1}{npT} \sum_{t = 1}^{T} (\xmat{W} \otimes \xmat{I})' \ddot{\xvec{Y}}_t'  \xvec{U}_t - \frac{1}{npT} E \left[ \sum_{t = 1}^{T}  (\xmat{W} \otimes \xmat{I})' \ddot{\xvec{Y}}_t'  \xvec{U}_t \right] \overset{p}{\to} 0, $ and
			\item[*] $\frac{1}{npT} \sum_{t = 1}^{T} \ddot{\xvec{Y}}_{t-1}'  \xvec{U}_t - \frac{1}{npT} E \left[ \sum_{t = 1}^{T}  \ddot{\xvec{Y}}_{t-1}'  \xvec{U}_t \right] \overset{p}{\to} 0$.
		\end{itemize}
		Because $\tilde{\xmat{A}} - \tilde{\xmat{A}}_0$ is uniformly bounded, the remaining terms converge to zero in probability by Chebycheff's inequality. Moreover, as $\xvec{\vartheta}  = (vec(\tilde{\xmat{A}})', vec(\xmat{\Psi})', vec(\xmat{\Pi})')'$ is bounded in $\Theta$, we get that
		\[ \frac{1}{npT} \sum_{t=1}^{T} \tilde{\xvec{U}}_t(\xi)'\tilde{\xvec{U}}_t(\xi) - \frac{1}{npT} E \left[\sum_{t=1}^{T} \tilde{\xvec{U}}_t(\xi)'\tilde{\xvec{U}}_t(\xi)\right] \overset{p}{\to} 0 \]
		uniformly in $\xvec{\vartheta} \in \Theta$, and, thus, 
		\[ \frac{1}{npT}\ln \mathcal{L}_{nT}(\tilde{\xmat{A}}, \xmat{\Psi}, \xmat{\Pi}|\xvec{Y}_0) - \frac{1}{npT} Q(\tilde{\xmat{A}}, \xmat{\Psi}, \xmat{\Pi}|\xvec{Y}_0) \overset{p}{\to} 0 \]
		uniformly in $\xvec{\vartheta} \in \Theta$. 
		
		Further, the equicontinuity of the expected likelihood must be shown. Let $\xmat{\iota}_{ij}$ be a zero matrix with the $(i,j)$-th entry equal to one. First, $\frac{1}{np} \frac{\partial \ln |\xmat{S}_{np}|}{\partial \psi_{ij}} = \frac{1}{np} tr(\xmat{S}_{np}^{'-1} (\xmat{\iota}_{ij} \otimes \xmat{W}))$ is uniformly bounded by a constant, uniformly in $\xmat{\Psi}$, because $\xmat{S}_{np}^{-1}$ is uniformly bounded according to Lemma \ref{lemma:UB_Psi}. Secondly, $\frac{1}{np}  \ln |\xmat{S}_{np}|$ is a Lipschitz function in $\xmat{\Psi}$ and, thus, uniformly equicontinuous. Thirdly, 
		\begin{eqnarray*}
			\sum_{t = 1}^{T} \tilde{\xvec{U}}_t(\xi)'\tilde{\xvec{U}}_t(\xi) & = & \left[\xmat{S}_{np} \xmat{S}_{np0}^{-1} vec(\tilde{\xmat{A}}_0 - \tilde{\xmat{A}}) +  \xmat{S}_{np} \xmat{S}_{np0}^{-1} (\xmat{\Pi}'_0 - \xmat{\Pi}')\ddot{\xvec{Y}}_{t-1}\right]' \\
			& & \quad \times \left[\xmat{S}_{np} \xmat{S}_{np0}^{-1} vec(\tilde{\xmat{A}}_0 - \tilde{\xmat{A}}) +  \xmat{S}_{np} \xmat{S}_{np0}^{-1} (\xmat{\Pi}'_0 - \xmat{\Pi}')\ddot{\xvec{Y}}_{t-1}\right]
		\end{eqnarray*}
		is uniformly equicontinuous, because $\tilde{\xmat{A}}$ and $\xmat{\Pi}$ are bounded, $\xmat{S}_{np}(\xmat{\Psi})$ is uniformly bounded in $\xmat{\Psi}$ and $\ddot{\xvec{Y}}_{t}' \ddot{\xvec{Y}}_{t}$ is $O(1)$ in $\xvec{\vartheta}$ according to Lemma \ref{lemma:conv}. Then, since $\xmat{S}_{np}^{-1}$ is $O(1)$ in $\xmat{\Psi}$ and $\xmat{S}_{np}(\xmat{\Psi})\xmat{S}_{np0}^{-1} = \xmat{I} - ((\xmat{\Psi}' - \xmat{\Psi}_0') \otimes \xmat{W}) \xmat{S}_{np0}^{-1}$, also
		\begin{equation*}
			\frac{1}{2np} tr\left( \xmat{S}_{np} \xmat{S}_{np0}^{-1} \xmat{S}_{np0}^{'-1} \xmat{S}_{np}' \right) =
			\frac{1}{2np} tr\left( (\xmat{I} - ((\xmat{\Psi}' - \xmat{\Psi}_0') \otimes \xmat{W}) \xmat{S}_{np0}^{-1}) (\xmat{I} - ((\xmat{\Psi}' - \xmat{\Psi}_0') \otimes \xmat{W}) \xmat{S}_{np0}^{-1})' \right) 
		\end{equation*}
		is a Lipschitz function in $\xmat{\Psi}$. Thus, this term is uniformly equicontinuous. 
		 
		 Because all terms are uniformly equicontinuous, also $\frac{1}{npT} Q(\tilde{\xmat{A}}, \xmat{\Psi}, \xmat{\Pi}|\xvec{Y}_0)$ is uniformly equicontinuous. 
	\end{itemize}
	Because $\xvec{\vartheta}_0$ is uniquely identified and the log-likelihood uniformly converges to the uniformly equicontinuous $\frac{1}{npT} Q(\tilde{\xmat{A}}, \xmat{\Psi}, \xmat{\Pi}|\xvec{Y}_0)$ in $\xvec{\vartheta} = (vec(\tilde{\xmat{A}})', vec(\xmat{\Psi})', vec(\xmat{\Pi})')'$, the consistency follows. This completes the proof.
\end{proof}

\end{appendix}

\end{document}